\documentclass[aps,pra,twocolumn,superscriptaddress,groupedaddress]{revtex4-1}  
\usepackage{graphicx}  
\usepackage{dcolumn}   
\usepackage{bm}        
\usepackage{amssymb}   

\usepackage{hhline}
\usepackage[caption=false]{subfig}
\usepackage{scalefnt}
\usepackage{amsmath}
\usepackage{pbox}
\usepackage{lipsum}
\usepackage{color}
\usepackage{xr}
\usepackage{cases}
\usepackage{enumerate}
\usepackage[caption=false]{subfig}

\usepackage{tikz}
\usetikzlibrary{matrix,arrows}

\usepackage{jabbrv}

\usepackage{enumitem}
\let\olddesc\description
\def\description{\olddesc\setlist[itemize]{leftmargin=*,labelindent=-12pt}}

\usepackage{lineno}

\usepackage{xr}

\usepackage[english]{babel}

\usepackage{graphicx}

\usepackage{xcolor}

\usepackage{tcolorbox}
\tcbuselibrary{breakable}
\tcbuselibrary{skins}
\tcbset{enhanced, colback=white, arc=0mm, outer arc=0mm, fonttitle=\bfseries, fontupper=\small}
\newtcolorbox[auto counter]{examplebox}[2][]{%
title=Box~\thetcbcounter: #2,#1}

\usepackage{amsmath}
\usepackage{amsthm}
\usepackage{amssymb}
\usepackage{mathtools}
\usepackage{thmtools}

\usepackage{bm}  
\usepackage{bbm}
\usepackage{dsfont}
\usepackage{cases}
\usepackage{nicefrac} 
\usepackage{xfrac}

\newtheorem{theorem}{Theorem}
\newtheorem{corollary}[theorem]{Corollary}
\newtheorem{prop}[theorem]{Proposition}

\newtheorem{definition}[theorem]{Definition}

\newcommand{\PN}{\textit{SciNet}}
\newcommand{\lab} {a_{\textnormal{corr}}}

\newcommand{\tin} {t}
\newcommand{\tout} {t'}
\newcommand{\toutPred} {\tout_{\textnormal{pred}}}

\newcommand{\thetaM} {\theta_M}
\newcommand{\thetaS} {\theta_S}
\newcommand{\phiE} {\phi_E}
\newcommand{\phiM} {\phi_M}

\newcommand{\DKL} {D_\textnormal{KL}}

\newcommand*{\Normal}{\ensuremath{\mathcal{N}}}
\DeclareMathOperator{\E}{\mathbb{E}}

\newcommand{\norm}[1]{\left\lVert#1\right\rVert}

\newcommand{\braket}[2]{\left\langle#1,#2\right\rangle}

\let\epsilon\relax
\def\epsilon{\varepsilon}

\def\id{\mathbbm 1}

\usepackage{slashbox}
 \usepackage{enumitem}
 \usepackage{relsize}

\newcommand\restr[2]{{
  \left.\kern-\nulldelimiterspace 
  #1 
  \vphantom{\big|} 
  \right|_{#2} 
  }}

\begin{document}
\title{Discovering physical concepts with neural networks} 

\author{Raban~Iten} 
 \thanks{These two authors contributed equally.}
\affiliation{ETH Z\"urich, Wolfgang-Pauli-Str. 27, 8093 Z{\"u}rich, Switzerland.} 
\author{ Tony Metger} 
 \thanks{These two authors contributed equally.}
\author{ Henrik Wilming}
\author{  L\'idia del Rio}
\author{ Renato Renner} 
\affiliation{ETH Z\"urich, Wolfgang-Pauli-Str. 27, 8093 Z{\"u}rich, Switzerland.}

\date{\today}

\begin{abstract} 
Despite the success of neural networks at solving concrete physics problems, their use as a general-purpose tool for scientific discovery is still in its infancy. 
Here, we approach this problem by modelling a neural network architecture after the human physical reasoning process, which has similarities to representation learning.
This allows us to make progress towards the long-term goal of machine-assisted scientific discovery from experimental data without making prior assumptions about the system. We apply this method to toy examples and show that the network finds the physically relevant parameters, exploits conservation laws to make predictions, and can help to gain conceptual insights, e.g.~Copernicus' conclusion that the solar system is heliocentric.
\end{abstract}

\maketitle

Theoretical physics, like all fields of human activity, is influenced by the schools of thought prevalent at the time of development. 
As such, the physical theories we know may not necessarily be the simplest ones to explain experimental data, but rather the ones that most naturally followed from a previous theory at the time. 
Both general relativity and quantum theory were built upon classical mechanics --- they have been impressively successful in the restricted regimes of the very large and very small, respectively, but are fundamentally incompatible, as reflected by paradoxes such as the black hole information loss~\cite{PhysRevD.14.2460, preskill_black_1992}.
This raises an interesting question: are the laws of quantum physics, and other physical theories more generally, the most natural ones to explain data from experiments if we assume no prior knowledge of physics? While this question will likely not be answered in the near future, recent advances in artificial intelligence allow us to make a first  step in this direction. 
Here, we investigate whether neural networks can be used to discover physical concepts from experimental data.

\textbf{Previous work.} 
The goal of using machines to help with discovering the physical laws underlying experimental data has been pursued in several contexts (see Appendix~\ref{sec:comparison} for a more detailed overview and~\cite{dunjko_machine_2017,roscher_explainable_2019,alhousseini_physicists_2019,carleo_machine_2019} for recent reviews). A lot of early work focused on finding mathematical expressions describing a given dataset (see e.g.~\cite{Crutchfield1987,Schmidt2009,schmidt_automated_2011}). For example, in~\cite{Schmidt2009} an algorithm recovers the  laws of motion of simple mechanical systems, like a double pendulum, by searching over a space of mathematical expressions on given input variables. More recently, significant progress was made in extracting dynamical equations from experimental data~\cite{daniels_automated_2015,brunton_discovering_2016,lusch_deep_2017,takeishi_learning_2017,otto_linearly-recurrent_2017,raissi_deep_2018,zhang_learning_2019,raissi_hidden_2018,raissi_physics-informed_2019}. These methods are highly practical and they were successfully applied to complex physical systems, but require prior knowledge on the systems of interest, for example in the form of knowing what the relevant variables are or that dynamics should be described by differential equations. In certain situations one might not have such prior knowledge or does not want to impose it to allow the machine to find entirely different representations of the physical system.

 \begin{figure*}[t] 
    \centering
    \subfloat[]{%
   \includegraphics[width=0.4\textwidth]{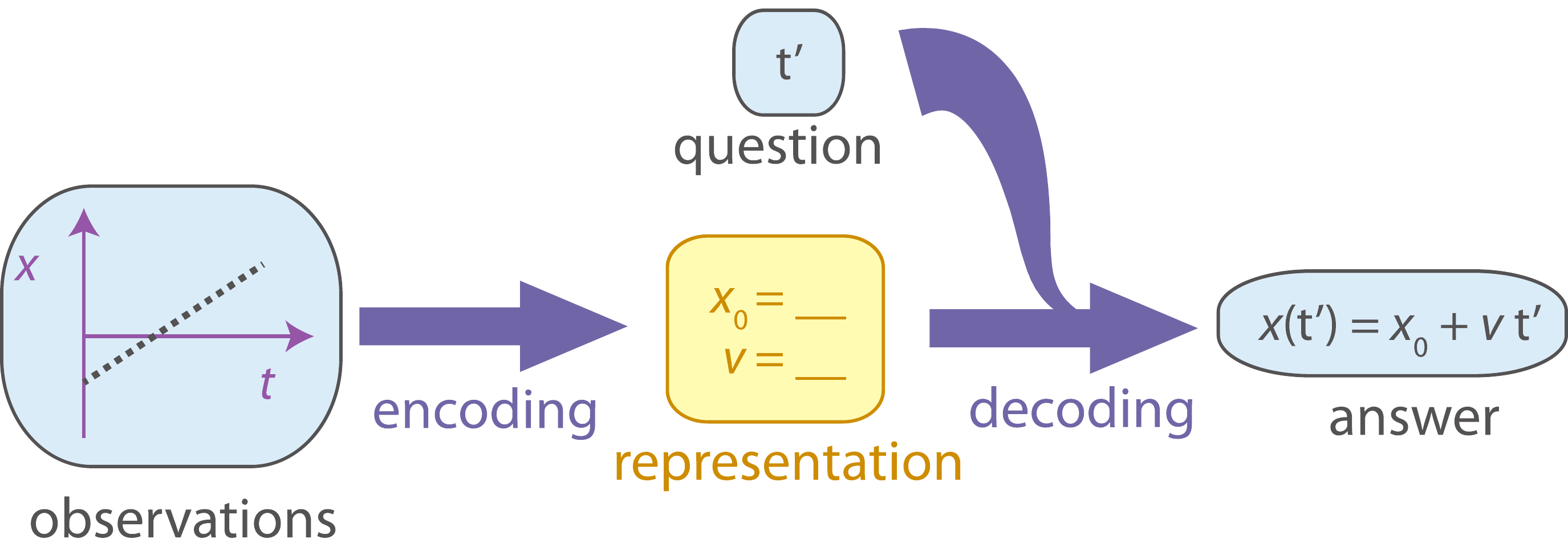} 
   }
   \quad \quad\quad \quad\quad \quad
    \subfloat[]{%
  \includegraphics[width=0.3\textwidth]{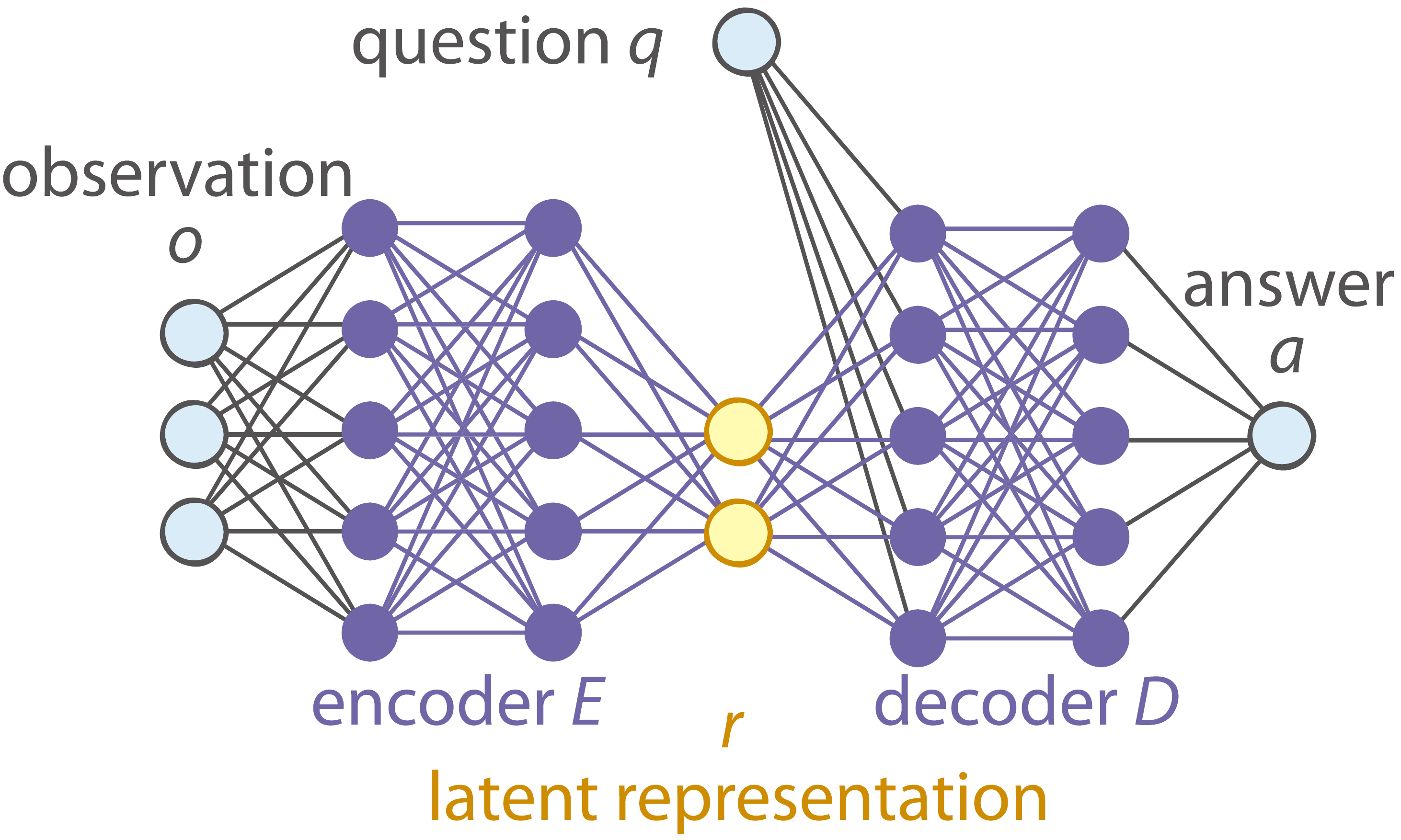}
  }
\caption{{\bf Learning physical representations.} {\bf (a) Human learning.} A physicist compresses experimental observations into a simple representation (\emph{encoding}). When later asked any question about the physical setting, the physicist should be able to produce a correct answer using only the representation and not the original data. We call the process of producing the answer from the representation \emph{decoding}. 
For example, the observations may be the first few seconds of the trajectory of a particle moving with constant speed; the representation could be the parameters ``speed $v$'' and ``initial position $x_0$'' and the question could be ``where will the particle be at a later time $t'$?'' 
{\bf (b) Neural network structure for \PN{}.} 
Observations are encoded as real parameters fed to an encoder (a \emph{feed-forward neural network}, see Appendix~\ref{app:NN}), which compresses the data into a representation (\emph{latent representation}). The question is also encoded in a number of real parameters, which, together with the representation, are fed to the decoder network to produce an answer. 
(The number of neurons depicted is not representative.)
  }
 \label{fig:scinet_general}
\end{figure*}

Over the last few years, neural networks have become the dominant method in machine learning and they have successfully been used to tackle complex problems in classical as well as quantum physics (see Appendix~\ref{sec:comparison} for further discussions). Conversely, neural networks may also lead to new insights into how the human brain develops physical intuition from observations~\cite{Bates_humans_2015,Wu_galileo_2015, Bramley2018,rempe_learning_2019,kissner_adding_2019,ehrhardt_unsupervised_2018,ye_interpretable_2018}. Very recently, physical variables were extracted in an unsupervised way from time series data of dynamical systems in~\cite{zheng_unsupervised_2018}.  

Our goal in this work is to minimize the extent to which prior assumptions about physical systems impose structure on the machine learning system. 
Eliminating assumptions that may not be satisfied for all physical systems, such as assuming that particles only interact in a pairwise manner, is necessary for the long-term goal of an artificial intelligence physicist (see~\cite{wu_toward_2018} for recent progress in this direction) that can be applied to any system without a need for adaptions and might eventually contribute to progress in the foundations of physics. 
Very recently,  neural networks were used in this spirit to detect differences between observed data and a reference model~\cite{de_simone_guiding_2019, dagnolo_learning_2019}.
However, there is a tradeoff between generality and performance, and the performance of the machine learning system proposed here --- based on autoencoders~\cite{Bengio2012,hinton_reducing_2006,Higgins2017} --- is not yet comparable to more established approaches that are adapted to specific physical systems.

\textbf{Modelling the physical reasoning process.} This work makes progress towards an interpretable artificial intelligence agent that is unbiased by prior knowledge about physics by proposing to focus on the human physical modelling process itself, rather than on specific physical systems. We formalize a simplified physical modelling process, which we then translate into a neural network architecture. This neural network architecture can be applied to a wide variety of physical systems, both classical and quantum, and is flexible enough to accommodate different additional desiderata on representations of the system that we may wish to impose.

We start by considering a simplified version of the physical modelling process, pictured in Figure~\ref{fig:scinet_general}a. Physicists' interactions with the physical world take the form of experimental observations (e.g.\ a time series $(t_i,x(t_i))_{i \in \{1,\dots,N\}}$ describing the motion of a particle at constant speed). The models physicists build do not deal with these observations directly, but rather with a representation of the underlying physical state of the observed system (e.g.\ the two parameters \emph{initial position} and \emph{speed}, $(x_0, v)$). Which parameters are used is an important part of the model, and we will give suggestions about what makes a good representation below. Finally, the model specifies how to make predictions (i.e., answer questions) based on the knowledge of the physical state of the system (e.g.\ ``where is the particle at time $t'$?'').
More formally, this physical modelling process can be regarded as an ``encoder'' $E: \mathcal O \to \mathcal R$ mapping the set of possible observations $\mathcal O$ to representations $\mathcal R$, followed by a  ``decoder''  $D: \mathcal R \times \mathcal Q \to \mathcal A$ mapping the sets of all possible representations $\mathcal R$ and questions $\mathcal Q$ to answers $\mathcal A$.

\textbf{Network structure.}  This modelling process can be translated directly into a neural network architecture, which we refer to as \PN{} in the following (Figure~\ref{fig:scinet_general}b). The encoder and decoder are both implemented as feed-forward neural networks. The resulting architecture, except for the question input, resembles an autoencoder in representation learning~\cite{Bengio2012,hinton_reducing_2006}, and more specifically the architecture in~\cite{eslami_neural_2018}. During the training, we provide triples of the form $(o,q,a_{\textnormal{corr}}(o,q))$ to the network, where $a_{\textnormal{corr}}(o,q) \in \mathcal A$ is the correct reply to question $q  \in \mathcal Q$ given the observation $o \in \mathcal O$. The learned parameterization is typically called latent representation~\cite{Bengio2012,hinton_reducing_2006}. To feed the questions into the neural network, they are encoded into a sequence of real parameters. Thereby, the actual representation of a single question is irrelevant as long as it allows the network to distinguish questions that require different answers.

It is crucial that the encoder is completely free to choose a latent representation itself, instead of us imposing a specific one.  Because neural networks with at least one hidden  layer composed of sufficiently many neurons can approximate any continuous function arbitrarily well~\cite{HORNIK1991251}, the fact that the functions $E$ and $D$ are implemented as neural networks does not significantly restrict their generality. However, unlike in an autoencoder, the latent representation need not describe the observations completely; instead, it only needs to contain the information necessary to answer the questions posed.

This architecture allows us to extract knowledge from the neural network: all of the useful information is stored in the representation, and the size of this representation is small compared to the total number of degrees of freedom of the network. This allows us to interpret the learnt representation. Specifically, we can compare \PN{}{\it 's} latent representation to a hypothesized parameterization to obtain a simple map from one to the other. If we do not even have any hypotheses about the system at hand, we may still gain some insights solely from the number of required parameters or from studying the change in the representation when manually changing the input, and the change in output when manually changing the representation (as in e.g.~\cite{Higgins2017}).

\begin{figure}
\centering
 \subfloat[One qubit\label{fig:one_qubit}]{%
  \includegraphics[width=.5\linewidth]{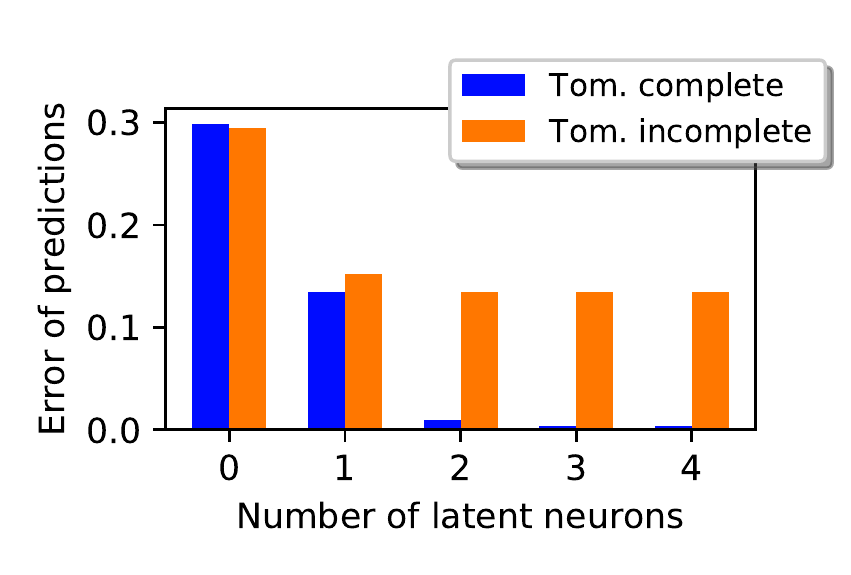}%
}\hfill
\subfloat[Two qubits\label{fig:two_qubit}]{%
  \includegraphics[width=.5\linewidth]{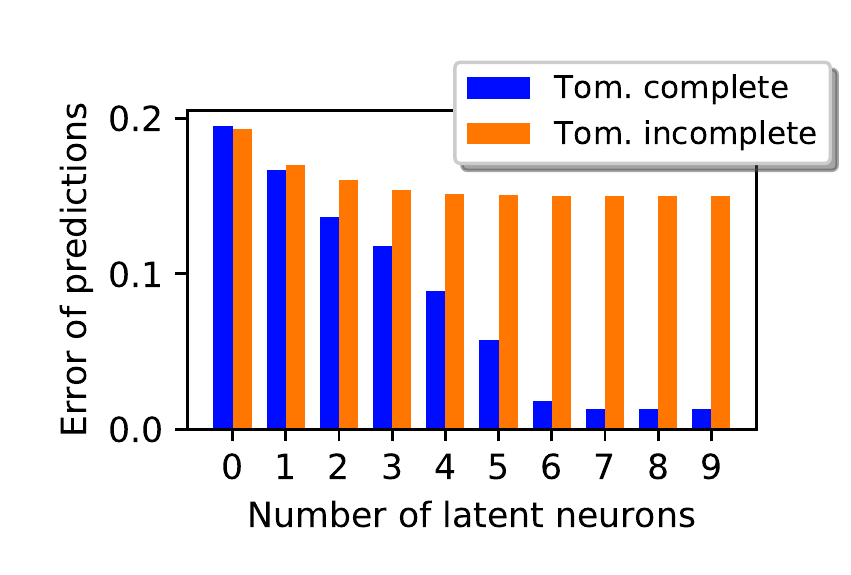}%
}
\caption{{\bf Quantum tomography.}  \PN{}  is given tomographic data for one or two qubits, as shown in part a) and b) of the figure, respectively, and an operational description of a measurement as a question input and has to predict the probabilities of outcomes for this measurement. The plots show the root mean square error of \PN{}\textit{'s} measurement predictions for test data as a function of the number of latent neurons. In the tomographically complete case, \PN{} recovers the number of (real) degrees of freedom required to describe a one and a two qubit state (which are two and six, respectively). Tomographically incomplete data can be recognized, since the prediction error remains high as one increases the number of latent neurons.}
\label{fig:qubits_tomography}
\end{figure}

\textbf{Desired properties for a representation.} For \PN{} to produce physically useful representations, we need to formalize what makes a good parameterization of a physical system, i.e., a good latent representation. We stress that this is not a property of a physical system, but a choice we have to make. We will give two possible choices  below. 

Generally, the latent representation should only store the minimal amount of information that is sufficient to correctly answer all questions in $\mathcal Q$. For \emph{minimal sufficient uncorrelated representations}, we additionally require that the latent neurons be statistically independent from each other for an input sampled at random from the training data, reflecting the idea that physically relevant parameters describe aspects of a system that can be varied independently and are therefore uncorrelated in the experimental data. Under this independence assumption, the network is then motivated to choose a representation that stores different physical parameters in different latent neurons. 
We formalize these demands in Appendix~\ref{sec:rep} and show, using techniques from differential geometry, that the number of latent neurons equals the number of underlying degrees of freedom in the training data  that are needed to answer all questions $\mathcal Q$. To implement these requirements in a neural network, we use well-established methods from representation learning, specifically disentangling variational autoencoders~\cite{Kingma2013,Higgins2017} (see Appendix~\ref{app:VAE} for details).

Alternatively, for situations where the physically relevant parameters can change, either over time or by some time-independent update rule, we might prefer a representation with a simple such update rule. We explain below how this requirement can be enforced.

\begin{figure*}[t] 
\centering
   \subfloat[\label{fig:NN_structure_for_dynamic_variables}]{%
  \includegraphics[height=3cm, width=7cm]{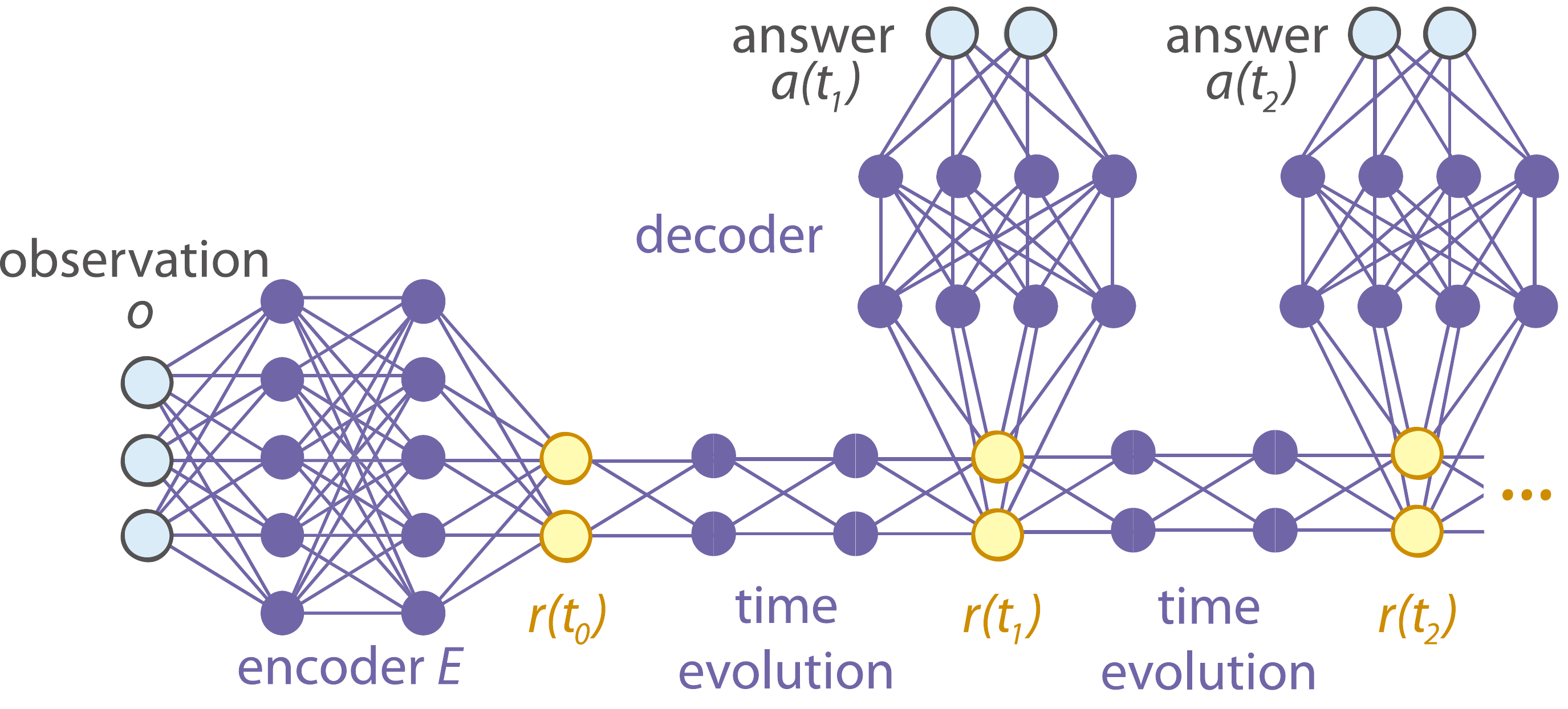}%
}\hfill
 \subfloat[\label{fig:copernicus}]{%
  \includegraphics[height=3cm, width=3.5cm]{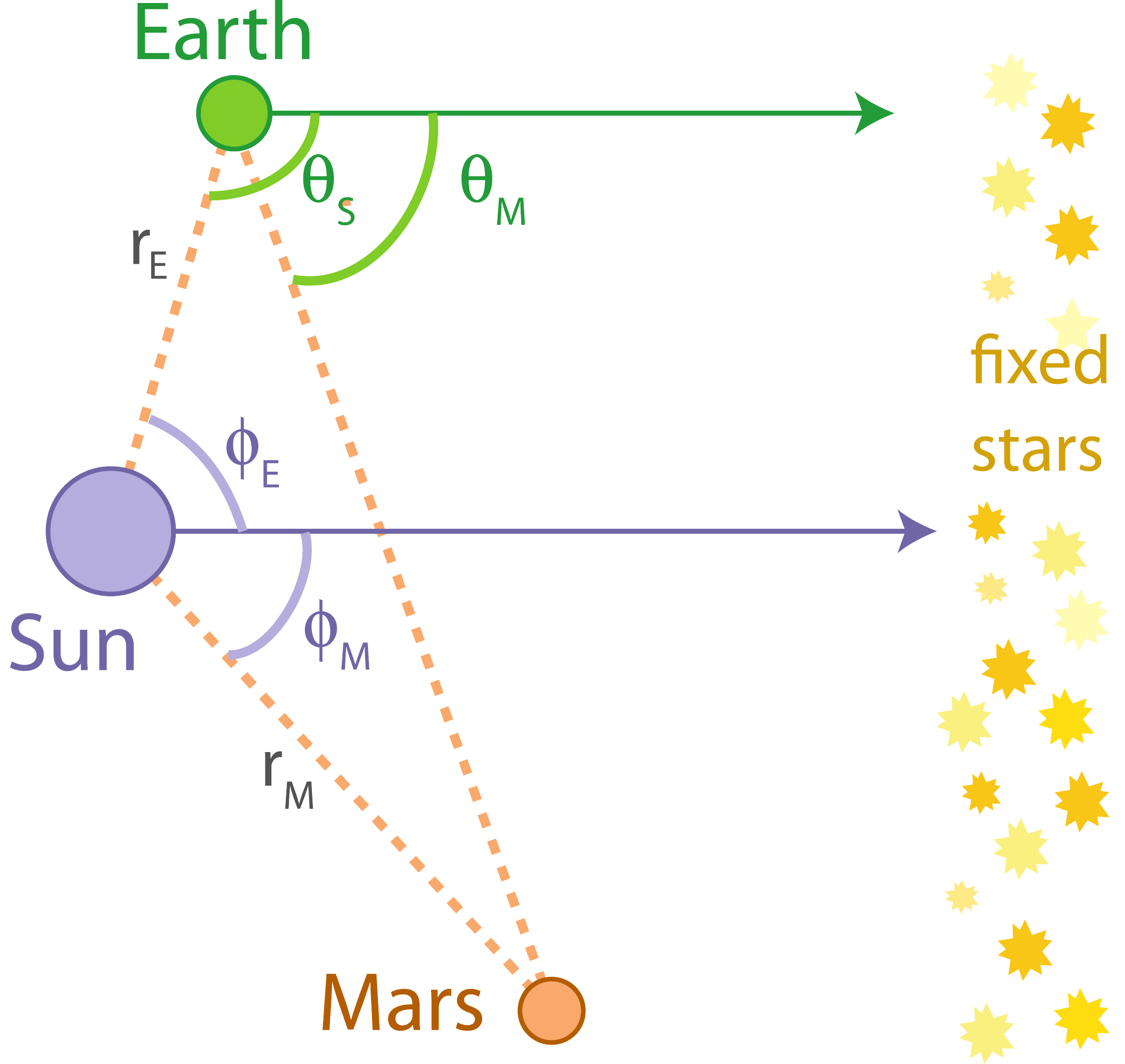}%
  }\hfill
\subfloat[ \label{fig:copernicus_phi_plot}]{%
  \includegraphics[height=3cm,width=6cm]{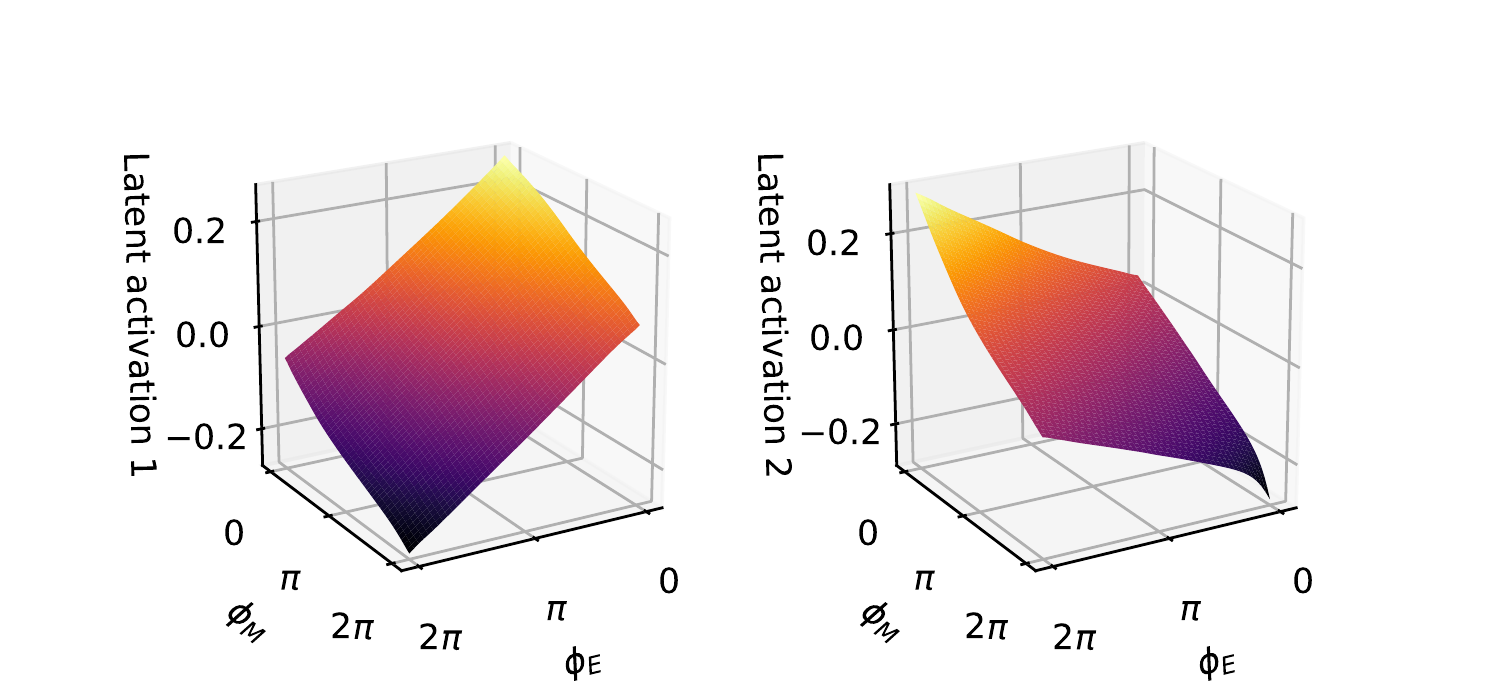}%
}
 \caption{{\bf Heliocentric model of the solar system.} \PN{} is given the angles of the Sun and Mars as seen from Earth at an initial time $t_0$ and has to predict these angles for later times.  
 {{\bf (a) Recurrent version of \PN{} for time-dependent variables.}  Observations are encoded into a simple representation $r(t_0)$ at time $t_0$. Then, the representation is evolved in time to $r(t_1)$ and a decoder is used to predict $a(t_1)$, and so on. In each (equally spaced) time step, the same time evolution network and decoder network are applied. }
 {\bf (b)~Physical setting.} The heliocentric angles $\phi_E$ and $\phi_M$ of the Earth and Mars are observed from the Sun; the angles $\theta_S$ and $\theta_M$ of the Sun and Mars are observed from Earth. All angles are measured relative to the fixed star background.
 {\bf (c)~Representation learned by \PN{}.} The activations $r_{1, 2}(t_0)$ of the two latent neurons at time $t_0$ (see Figure~\ref{fig:NN_structure_for_dynamic_variables}) are plotted as a function of the heliocentric angles $\phiE$ and $\phiM$. The plots show that the network stores and evolves parameters that are linear combinations of the heliocentric angles.} 
 \label{fig:copernicus_all}
\end{figure*}

\textbf{Results.} To demonstrate that \PN{} helps to recover relevant concepts in physics by providing the relevant physical variables, both in quantum- and classical-mechanical settings, we consider four toy examples from different areas of physics. In summary, we find:
(i)~given a time series of the positions of a damped pendulum, \PN{} can predict future positions with high accuracy and it uses the relevant parameters, namely frequency and damping factor, separately in two of the latent neurons (and sets the activation of unnecessary latent neurons to zero); 
(ii)~\PN{} finds and exploits conservation laws: it uses the total angular momentum to predict the motion of two colliding particles;
(iii)~given measurement data from a simple quantum experiment, \PN{} can be used to determine the dimension of the underlying unknown quantum system and to decide whether a set of measurements is tomographically complete, i.e., whether it provides full information about the quantum state;
(iv)~given a time series of the positions of the Sun and Mars as observed from Earth, \PN{} switches to a heliocentric representation --- that is, it encodes the data into the angles of the two planets as seen from the Sun.
The results show that \PN{} finds, without having been given any prior information about the specific physical systems, the same quantities that we use in physics textbooks to describe the different settings. We also show that our results are robust against noise in the experimental data.
To illustrate our approach, we will now describe two of these settings in some depth. For detailed descriptions of the four different settings, the data generation, interpretation and additional background information, we refer to Appendix~\ref{app:sec:examples}. 

In all our examples, the training data we use is operational and could be generated from experiments, i.e., the correct answer is the one observed experimentally. Here, we use simulations instead because we only deal with classical and quantum mechanics, theories whose predictions are experimentally well tested in the relevant regimes. One might think that using simulated data would restrict \PN{} to rediscovering the theory used for data generation. However, in particular for quantum mechanics, we are interested in finding conceptually different formulations of the theory with the same predictions.

\textbf{Quantum state tomography.} In quantum mechanics, it is not trivial to construct a simple representation of the state of a quantum system from measurement data, a task called quantum tomography~\cite{paris_quantum_2004}. In the following, we will show that \PN{} finds representations of arbitrary (pure) one- and two-qubit states. To ensure that no prior knowledge about quantum physics is required to collect the measurement data, we assume an operational setting in which we have access to two devices in a lab, where one device can create (many copies of) a quantum system in a certain state depending on the chosen parameters of the device. The other device performs binary measurements on the quantum system. The input to \PN{} consists of the outcome probabilities of a random fixed set of ``reference measurements'' on quantum systems in the unknown quantum state $\psi$. As a question input, we provide a parameterization of a measurement $\omega$ (one may think of the setting of the dials and buttons of the measurement device). \PN{} has to predict the outcome probability of the measurement $\omega$ on a quantum system in the state $\psi$. We train \PN{} with different pairs $(\omega,\psi)$ for one and two qubits.
The results are shown in Figure~\ref{fig:qubits_tomography}. Training different networks with different numbers of latent neurons, we can observe how the quality of the predictions (after training has been completed) improves as we allow for more parameters in the representation of $\psi$. This allows us to gain relevant information, without previous hypotheses about the nature of this representation (for example, whether it is a vector in a Hilbert space).

If the reference measurements are tomographically complete, meaning that they are sufficient to reconstruct a complete representation of the underlying quantum system, the plots in Figure~\ref{fig:qubits_tomography} show a drop in prediction error when the number of latent neurons is increased up to two and six for the cases of one and two qubits, respectively~\footnote{In the case of a single qubit, there is an additional small improvement in going from two to three latent neurons: 
this is a technical issue caused by the fact that (finite size) neural networks cannot represent discontinuous functions (see Appendix~\ref{app:cyclic_rep}). The same likely applies in the case of two qubits, going from 6 to 7 latent neurons.}.
This is in accordance with the number of degrees of freedom required to describe a one- or a two-qubit state in our current theory of quantum mechanics.
For the case where the set of measurements is tomographically incomplete, it is not possible for \PN{} to predict the outcome of the final measurement perfectly regardless of the number of latent neurons. This means that purely from operational data, we can make a statement about the tomographic completeness of measurements and about the number of degrees of freedom of the underlying unknown quantum system.


\textbf{Enforcing a simple time evolution.} As mentioned above, if the physically relevant parameters can change, we can enforce a representation that has a simple update rule. For illustration, we will consider time evolution here, but more general update rules are possible. To accomodate changing physical parameters, we need to extend the latent representation as shown in Figure~\ref{fig:NN_structure_for_dynamic_variables}. Instead of a single latent represetation with a decoder attached to it, we now have many latent representations that are generated from the intial representation by a time evolution network. Each representation has a decoder attached to it to produce an answer to a question. Because we only want the parameters, but not the physical model, to change in time, all time evolution steps and decoders are identical, i.e., they implement the same function. The encoder, time evolution network, and decoder are trained simultaneously. To enforce parameters with a simple time evolution, we restrict the time evolution network to implementing very simple functions, such as addition of a constant~\footnote{For a general system, there might not exist a representation that admits such a simple time evolution. In such a case, one may have to define a complexity measure for the time update rule and search over different rules successively increasing the complexity until \PN{} can achieve good prediction accuracy.}. 

\textbf{Heliocentric solar system.}
In the 16th century, Copernicus used observations of the positions of different planets on the night sky (Figure~\ref{fig:copernicus}) to hypothesize that the Sun, and not the Earth, is at the centre of our solar system. This heliocentric view was confirmed by Kepler at the start of the 17th century  based on astronomic data collected by Brahe, showing that the planets move around the Sun in simple orbits. Here, we show that \PN{} similarly uses heliocentric angles when forced to find a representation for which the time evolution of the variables takes a very simple form, a typical requirement for time-dependent variables in physics.

The observations given to \PN{} are angles $\theta_M(t_0)$ of Mars and $\theta_S(t_0)$ of the Sun as seen from Earth at a starting time $t_0$ (which is varied during training). The time evolution network is restricted to addition of a constant (the value of which is learned during training).
At each time step $i$, \PN{} is asked to predict the angles as seen from Earth at the time $t_i$ using only its representation $r(t_i)$. Because this question is constant, we do not need to feed it to the decoder explicitly.

We train \PN{} with randomly chosen subsequences of weekly (simulated) observations of the angles $\thetaM$ and $\thetaS$  within Copernicus' lifetime (3665 observations in total).  For our simulation, we assume circular orbits of Mars and Earth around the Sun.
Figure~\ref{fig:copernicus_phi_plot} shows the learned representation and confirms that \PN{} indeed stores a linear combination of \emph{heliocentric angles}. We stress that the training data only contains angles observed from Earth, but \PN{} nonetheless switches to a heliocentric representation.

\textbf{Conclusion.} In this work, we have  shown that \PN{} can be used to recover physical variables from experimental data in various physical toy settings. 
The learnt representations turned out to be the ones commonly used in physics textbooks, under the assumption of uncorrelated sampling. In future work we plan to extend our approach to data where the natural underlying parameters are correlated in the training distribution. The separation of these parameters in the representation found by \PN{}  requires the development of further operational criteria for disentangling latent variables.
In more complex scenarios, the methods introduced here may lead to entirely novel representations, and extracting human physical insight from such representations remains challenging.
This could be addressed using methods from symbolic regression~\cite{Koza1994} to obtain analytical expressions for the encoder and decoder maps, or for a map between a hypothesized and the actual representation. Alternatively, methods such as the ones presented in~\cite{krenn_selfies:_2019, beny_learning_2019} could help to improve the interpretability of the representation.
Following this direction, it might eventually become possible for neural networks to produce insights expressed in our mathematical language.

\setlength{\parskip}{3pt}

 \textbf{Acknowledgments.} We would like to thank Alessandro Achille, Serguei Beloussov, Ulrich Eberle, Thomas Frerix, Viktor Gal, Thomas H\"aner, Maciej Koch-Janusz, Aurelien Lucchi,  Ilya Nemenman, Joseph M. Renes, Andrea Rocchetto,  Norman Sieroka, Ernest Y.-Z. Tan, Jinzhao Wang and Leonard Wossnig  for helpful discussions. We acknowledge support from the Swiss
 National Science Foundation through  SNSF project No.\ $200020\_165843$,
 the Swiss National Supercomputing Centre (CSCS) under project ID da04,
  and through the National Centre of
 Competence in Research \emph{Quantum Science and Technology} (QSIT). 
 LdR and RR furthermore acknowledge support from  the FQXi grant \emph{Physics of
 the observer}.
 TM acknowledges support from ETH Z{\"u}rich and the ETH Foundation through the \emph{Excellence Scholarship \& Opportunity Programme}. \\
 
 The source code and the training data are available at:
 \url{https://github.com/eth-nn-physics/nn_physical_concepts}. See also Appendix~\ref{app:impl_details} for the implementation details. \PN{} worked well on all tested examples, i.e., we did not post-select examples based on whether \PN{} worked or not.\\ 
\appendix
\section{Implementation} \label{app:impl_details}
The neural networks used in this work are specified by the input size for question and observation input, the output size, the number of latent neurons, and the sizes of encoder and decoder. The number of neurons in the encoder and decoder are not expected to be important for the results, provided the encoder and decoder are large enough to not constrain the expressivity of the network (and small enough to be efficiently trainable).
The training is specified by the number of training examples, the batch size, the number of epochs, the learning rate and the value of the parameter $\beta$ (see section \ref{app:VAE}). To test the networks, we use a number of previously unseen test samples.
We give the values of these parameters for the examples presented in the main text (and described in detail in Section \ref{app:sec:examples}) in Table~\ref{Table:network_spec} and Table~\ref{Table:training_details}.

The source code, all details about the network structure and training process, and  pre-trained \PN{}{\textit s} are available at

\url{https://github.com/eth-nn-physics/nn_physical_concepts}. 

\noindent The networks were implemented using the Tensorflow library~\cite{tensorflow2015-whitepaper}. For all examples, the training process only takes a few hours on a standard laptop. 

\begin{table*}[!t] 
\renewcommand{\arraystretch}{1.4}
\centering
\begin{tabular}{lllllll}
Example&\pbox{10cm}{Observation \\ input size}&\pbox{10cm}{Question \\ input size} & \pbox{10cm}{\# Latent \\ neurons} &\pbox{10cm}{Output \\ size} &Encoder&Decoder	\\ \hline 
Pendulum & 50 & 1 & 3 & 1 & $[500,100]$ & $[100,100]$\\
Collision & 30 & 16 & 1 & 2 & $[150,100]$ & $[100,150]$\\
One qubit & 10 & 10 & 0-5 & 1 & $[100,100]$ & $[100,100]$\\
Two qubits & 30 & 30 & 0-9 & 1 & $[300,100]$ & $[100,100]$\\
Solar system & 2 & 0 & 2 & 2 & $[100,100]$ & $[100,100]$\\  								   
\end{tabular}

\caption{Parameters specifying the network structure. The first four examples use the network structure depicted in 
Figure~\ref{fig:scinet_general}, whereas the last example uses the structure shown in Figure~\ref{fig:NN_structure_for_dynamic_variables}. The notation $[n_1,n_2]$ is used to describe the number of neurons $n_1$ and $n_2$ in the first and the second hidden layer of the encoder (or the decoder), respectively.}
\label{Table:network_spec}

\end{table*}

\begin{table*}[!t] 
\renewcommand{\arraystretch}{1.4}
\centering
\begin{tabular}{lllllll}
Example & Batch size & Learning rate&$\beta$ &\# Epochs & 	\pbox{10cm}{\# Training \\ samples} & \pbox{10cm}{\# Test \\ samples} \\ \hline 
Pendulum & $512$ & $10^{-3}$ & $10^{-3}$ & $1000$ & $95000$ & $5000$\\
Collision & $500$ & $(5\cdot10^{-4},10^{-4})$ & $0$  & $(100,50)$ & $490000$ & $10000$\\
One qubit & $512$ & $(10^{-3},10^{-4})$ & $10^{-4}$ & $(250,50)$ & $95000$ & $5000$\\
Two qubits & $512$ & $(10^{-3},10^{-4})$ & $10^{-4}$ & $(250,50)$ & $490000$ & $10000$\\
Solar system & $256-2048$ & $10^{-5}-10^{-4}$ & $0.001-0.1$ & $15000$ & $95000$ & $5000$\\ 									   
\end{tabular}

\caption{Parameters specifying the training process. For training with two phases, the notation $(p_1, p_2)$ refers to the parameters in the first and second phase, respectively. The last example uses five training phases specified in detail in \cite{copernicus_details}.}
\label{Table:training_details}

\end{table*}

\section{Detailed comparison with previous work} \label{sec:comparison} 

Neural networks have become a standard tool to tackle problems where we want to make predictions without following a particular algorithm or imposing structure on the available data (see for example~\cite{nielsenneural,lecun_deep_2015,silver_mastering_2016}) and they have been applied to a wide variety of problems in physics.
For example, in condensed matter physics and generally in many-body settings, neural networks have proven particularly useful to characterize phase transitions (see~\cite{dunjko_machine_2017} and references therein) and to learn local symmetries~\cite{decelle_learning_2019}.

In quantum optics, automated search techniques and reinforcement-learning based schemes have been used to generate new experimental setups~\cite{krenn_automated_2016,melnikov_active_2018}. 
Projective simulation~\cite{briegel_projective_2012}  is used in~\cite{melnikov_active_2018} to autonomously discover experimental building blocks with maximum versatility.

Closer to our work, neural networks have also been used to efficiently represent wave functions of particular quantum systems~\cite{carleo_solving_2017,rocchetto_learning_2018,glasser_neural-network_2018,carleo_constructing_2018, cai_approximating_2018,huang_neural_2017,deng_machine_2017,schmitt_quantum_2018,torlai_many-body_2018,nomura_restricted-boltzmann-machine_2017,deng_quantum_2017,gao_efficient_2017,carrasquilla_reconstructing_2019,beach_qucumber:_2019,torlai_machine_2019}. In particular, in~\cite{rocchetto_learning_2018}, variational autoencoders are used to approximate the distribution of the measurement outcomes of a specific quantum state for a fixed measurement basis and the size of the neural network can provide an estimate for the complexity of the state. In contrast, our approach does not focus on an efficient representation of a given quantum state and it is not specifically designed for learning representations of quantum systems. 
Nevertheless, \PN{} can be used to produce representations of arbitrary states of simple quantum systems without retraining.
This allows us to extract information about the degrees of freedom required to represent any state of a (small) quantum system.

Another step towards extracting physical knowledge in an unsupervised way is presented in~\cite{koch-janusz_mutual_2018}. The authors show how the relevant degrees of freedom of a system in classical statistical mechanics can be extracted under the assumption that the input is drawn from a Boltzmann distribution. 
They make use of information theory to guide the unsupervised training of restricted Boltzmann machines, a class of probabilistic neural networks, to approximate probability distributions.

A different line of work has focused on using neural networks and other algorithmic techniques to better understand how humans are able to gain an intuitive understanding of physics~\cite{hamrick_internal_2011,battaglia_simulation_2013,Bates_humans_2015,Wu_galileo_2015, Ullman2018, Bramley2018,rempe_learning_2019,kissner_adding_2019,ehrhardt_unsupervised_2018,ye_interpretable_2018}.

Very recently, physical variables were extracted in an unsupervised way from time series data of dynamical systems~\cite{zheng_unsupervised_2018}. 
The network structure used in~\cite{zheng_unsupervised_2018} is built on interaction networks~\cite{battaglia_interaction_2016,chang_compositional_2016,raposo_discovering_2017} and it is well adapted to physical systems consisting of several objects interacting in a pair-wise manner.
The prior knowledge included in the network structure allows the network to generalise to situations that differ substantially from those seen during training. 

In the last few years, significant progress was made in extracting dynamical equations from experimental data~\cite{daniels_automated_2015,brunton_discovering_2016,lusch_deep_2017,takeishi_learning_2017,otto_linearly-recurrent_2017,raissi_deep_2018,zhang_learning_2019,raissi_hidden_2018,raissi_physics-informed_2019}, which is known to be an NP-hard problem~\cite{cubitt_extracting_2012}.
Where the most of these works search for dynamical models in the input data, in~\cite{lusch_deep_2017,takeishi_learning_2017,otto_linearly-recurrent_2017} neural networks are used to find a new set of variables such that the evolution of the new variables is approximately linear (motivated by Koopman operator theory). 
Our example given in Section~\ref{sec:copernicus} uses a similar network structure as the one used in~\cite{lusch_deep_2017,takeishi_learning_2017,otto_linearly-recurrent_2017}, which corresponds to a special case of \PN{} with a trivial question input and a latent representation that is evolved in time. 
The concept of evolving the system in the latent representation has also been used in machine learning to extract the relevant features from video data~\cite{fraccaro_disentangled_2017}.
A further step towards an artificial intelligence physicist was taken in~\cite{wu_toward_2018}, where data from complex environments is automatically separated into parts corresponding to systems that can be explained by simple physical laws. The machine learning system then tries to unify the underlying ``theories'' found for the different parts of the data.

\section{Minimal representations} \label{sec:rep}
Here, we describe some of the theoretical considerations that went into designing \PN{} and helping it to find useful representations that encode physical principles. 
Given a data set, it is generally a complex task to find a simple representation of the data that contains all the desired information.
\PN{} should recover such representations by itself;
however, we encourage it to learn ``simple'' representations during training.
To do so, we have to specify the desired properties of a representation.
In this, our approach follows the spirit of several works on representation learning theory~\cite{Kingma2013,hinton_reducing_2006,Bengio2012,bengio_deep_2013,Higgins2017,Achille2018}. 

For the theoretical analysis, we introduce some additional structure on the data that is required to formulate the desired properties of a representation.  
We consider real-valued data, which we think of as being sampled from some unknown probability distribution. 
In other words, we assign random variables to the observations $O$, the questions $Q$, the latent representation $R$, and the answers $A$. We use the convention that a random variable $X=(X_1,\dots,X_{|X|})$ takes samples in  $\mathcal{X} \subset \mathbb{R}^{|X|}$, where $|X|$ denotes the dimension of the ambient space of $X$. In particular, $|R|$ will correspond to the number of neurons in the latent representation.

We require the following properties for an \textit{uncorrelated (sufficient) representation} $R$ (defined by an encoder mapping $E:\mathcal{O} \rightarrow \mathcal{R}$) for the data described by the triple ($O,Q,\lab)$, where we recall that the function $\lab: \mathcal{O}\times\mathcal{Q}\rightarrow\mathcal{A}$ sends an observation $o \in \mathcal{O}$ and a question $q \in \mathcal{Q}$ to the correct answer $a \in \mathcal{A}$.

\begin{enumerate}
\item \label{prop:sufficient} \textbf{Sufficient (with smooth decoder)}:
There exists a smooth map $D:\mathcal{R} \times \mathcal{Q} \mapsto \mathcal{A} $, such that $D(E(o),q)=\lab (o,q)$ for all possible observations $o \in \mathcal{O}$ and questions $q \in \mathcal{Q}\, .$
 \item \label{prop:uncorrelated}   \textbf{Uncorrelated}: The elements in the set $\{R_1,R_2,\dots,R_{|R|}\}$ are mutually independent.
\end{enumerate}
Property~\ref{prop:sufficient} asserts that  the encoder map $E$ encodes all information of the observation $o \in \mathcal{O}$  that is necessary to reply to all possible questions $q \in \mathcal{Q}$. We require the decoder to be smooth, since this allows us to give the number of parameters stored in the latent representation a well defined meaning in terms of a dimension (see Section~\ref{app:uniquess}).

Property~\ref{prop:uncorrelated} means that knowing some variables in the latent representation does not provide any information about any other latent variables; note that this depends on the distribution of the observations.

We define a {\it minimal uncorrelated representation} $R$ as an uncorrelated (sufficient) representation with a minimal number of parameters $|R|$.
This formalizes what we consider to be a ``simple'' representation of physical data.

Without the assumption that the decoder is smooth, it would, in principle, always be sufficient to have a single latent variable, since a real number can store an infinite amount of information. Hence, methods from standard information theory, like the information bottleneck~\cite{tishby_information_2000,tishby_deep_2015,shwartz-ziv_opening_2017}, are not the right tool to give the number of variables a formal meaning.
In Section~\ref{app:uniquess}, we use methods from differential geometry to show that the number of variables $|R|$ in a minimal (sufficient) representation corresponds to the number of relevant degrees of freedom in the observation data required to answer all possible questions.

\subsection{Interpretation of the number of latent variables} \label{app:uniquess}

Above, we have required that the latent representation should contain a minimal amount of latent variables; we now relate this number to the structure of the given data. Proposition~\ref{prop:latent_number} below asserts that the minimal number of latent neurons corresponds to the relevant degrees of freedom in the observed data required to answer all the questions that may be asked.

For simplicity,  we describe the data with sets instead of random variables here. Note that the probabilistic structure was only used for Property~\ref{prop:uncorrelated} in Section~\ref{sec:rep}, whereas here, we are only interested in the number of latent neurons and not in that they are mutually independent. We therefore consider the triple $(\mathcal{O},\mathcal{Q},\lab)$, where  $\mathcal{O}$ and  $\mathcal{Q}$ are the sets containing the observation data and the questions respectively, and the function $\lab:(o,q)\mapsto a$ sends an observation $o \in \mathcal{O}$ and a question $q \in \mathcal{Q}$ to the correct reply $a \in \mathcal{A}$.

Intuitively, we say that the triple $(\mathcal{O},\mathcal{Q},\lab)$ has dimension at least $n$ if there exist questions  in $\mathcal Q$ that are able to capture $n$ degrees of freedom from the observation data $\mathcal O$. Smoothness of this ``oracle'' is a natural requirement, in the sense that we expect the dependence of the answers on the input to be robust under small perturbations. The formal definition follows. 

\begin{definition}[Dimension of a data set] \label{defi:data_dim}
Consider a data set described by the triple $(\mathcal O, \mathcal Q, \lab)$, where $\lab: \mathcal O \times \mathcal Q \to \mathcal A$, and all sets are real, $\mathcal{O} \subseteq \mathbb{R}^r,\mathcal{Q}\subseteq \mathbb{R}^s, \mathcal{A} \subseteq \mathbb{R}^t $. 
We say that  this triple has  dimension at least $n$ if there exists an $n$-dimensional submanifold 
$\mathcal{O}_n\subseteq \mathcal{O}$ and  questions $q_1,\dots,q_k \in \mathcal Q$ and a function  
\begin{align*}
f: \quad \mathcal{O}_n  &\to \mathcal A^{k}:=\overbrace{A \times A \times \dots \times A}^k\\
o &\mapsto [\lab(o,q_1),\dots,\lab(o,q_k)] 
\end{align*}
such that $f: \mathcal{O}_n  \to f(\mathcal{O}_n)$  is a diffeomorphism.
\end{definition}

\begin{prop} [Minimal representation for \PN{}] \label{prop:latent_number}
A (sufficient) latent representation for data described by a triple $(\mathcal{O} \subset \mathbb{R}^r,\mathcal{Q}\subset \mathbb{R}^s,\lab:\mathcal{O} \times \mathcal{Q} \rightarrow \mathcal{A} \subset \mathbb{R}^t)$ of dimension at least $n$ requires at least $n$ latent variables.
\end{prop}
\begin{proof}
By assumption, there is an  $n$-dimensional submanifold $\mathcal{O}_n\subset \mathcal{O}$ and $k$ questions $q_1,\dots,q_k$ such that $f: \mathcal{O}_n  \to \mathcal I_n :=f(\mathcal{O}_n)$ is a diffeomorphism. 
We prove the statement by contradiction: assume that there exists a (sufficient) representation described by an encoder $E:\mathcal{O}\rightarrow \mathcal{R}_m\subset \mathbb{R}^m$ with $m<n$ latent variables.
By sufficiency of the representation, there exists a smooth decoder $D:\mathcal{R}_m \times \mathcal{Q}\rightarrow \mathcal{A}$ such that $D(E(o),q)=\lab(o,q)$ for all observations $o\in \mathcal{O}$ and questions $q \in \mathcal{Q}$. 
We define the smooth map 
\begin{align*}
\tilde{D}: \mathcal{R}_m &\to\mathcal A^{k} \\ 
r &\mapsto [D(r,q_1),\dots,D(r,q_k)], 
\end{align*}
and denote the pre-image of $\mathcal{I}_n$ by $\tilde{\mathcal{R}}_m:=\tilde{D}^{-1}(\mathcal{I}_n)$.

By sufficiency of the representation, the restriction of the map $\tilde{D}$ to $\tilde{\mathcal{R}}_m$ denoted by $\tilde{D}|_{\tilde{\mathcal{R}}_m}:\tilde{\mathcal{R}} _m\rightarrow \mathcal{I}_n$ is a smooth and surjective map. 
However, by Sard's theorem (see for example~\cite{lee_introduction_2012}), the image $\tilde{D}({\tilde{\mathcal{R}}}_m)$ is of measure zero in $\mathcal{I}_n$, since the dimension of the domain ${\tilde{\mathcal{R}}}_m \subset \mathbb{R}^m$ is at most $m$, which is smaller than the dimension $n$ of the image $\mathcal{I}_n$. This contradicts the surjectivity of $\tilde{D}|_{\tilde{\mathcal{R}}_m}$ and finishes the proof. 
\end{proof}

\begin{figure*}
\centering
 \subfloat[ \label{fig:neuron}]{%
  \includegraphics[width=.4\linewidth]{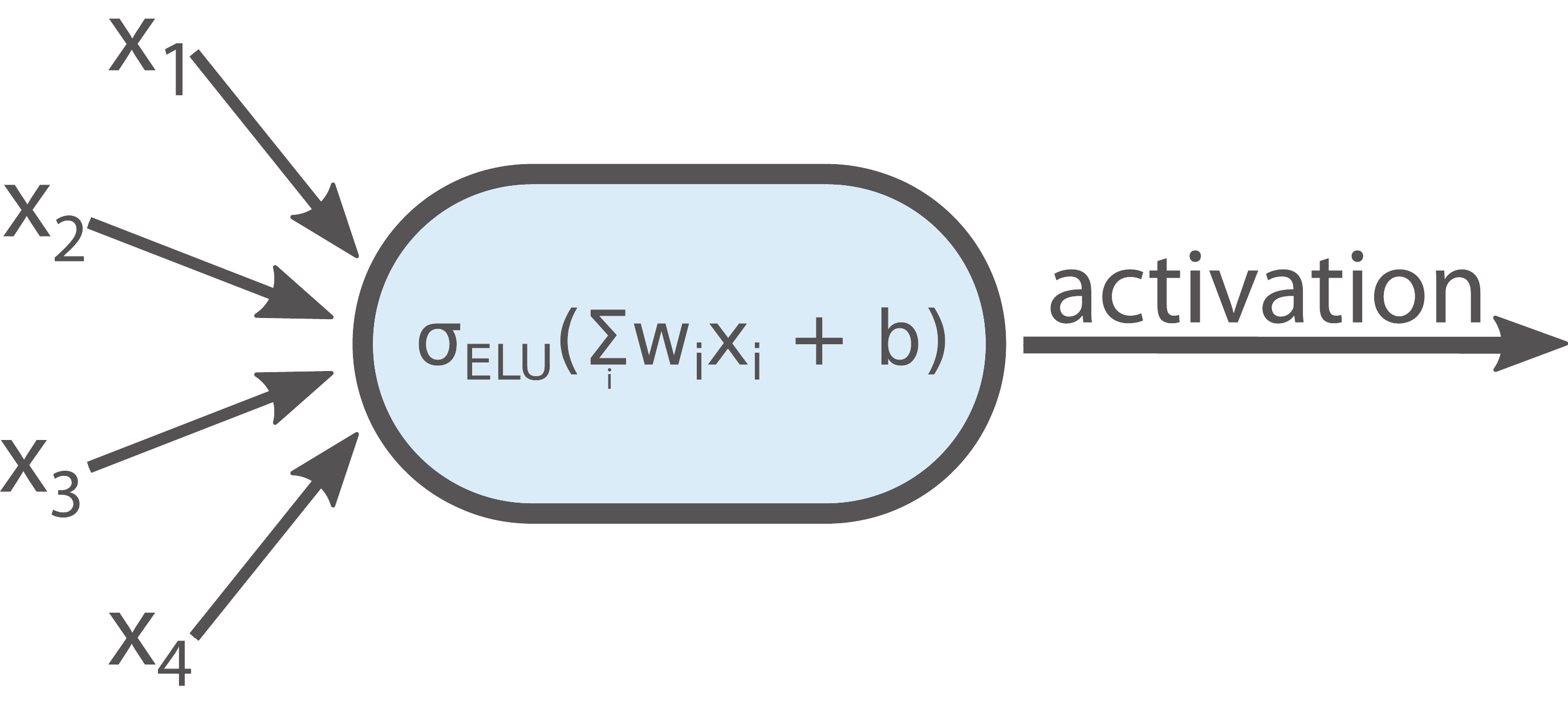}%
}\hfill
\subfloat[ \label{fig:ELU}]{%
  \includegraphics[width=.4\linewidth]{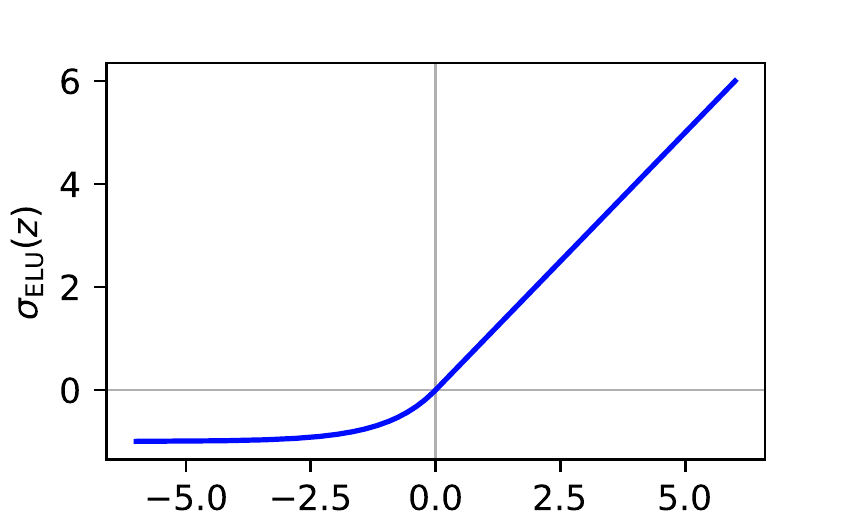}%
}\hfill
\subfloat[    \label{fig:feedforward_net}]{%
  \includegraphics[width=.35\linewidth]{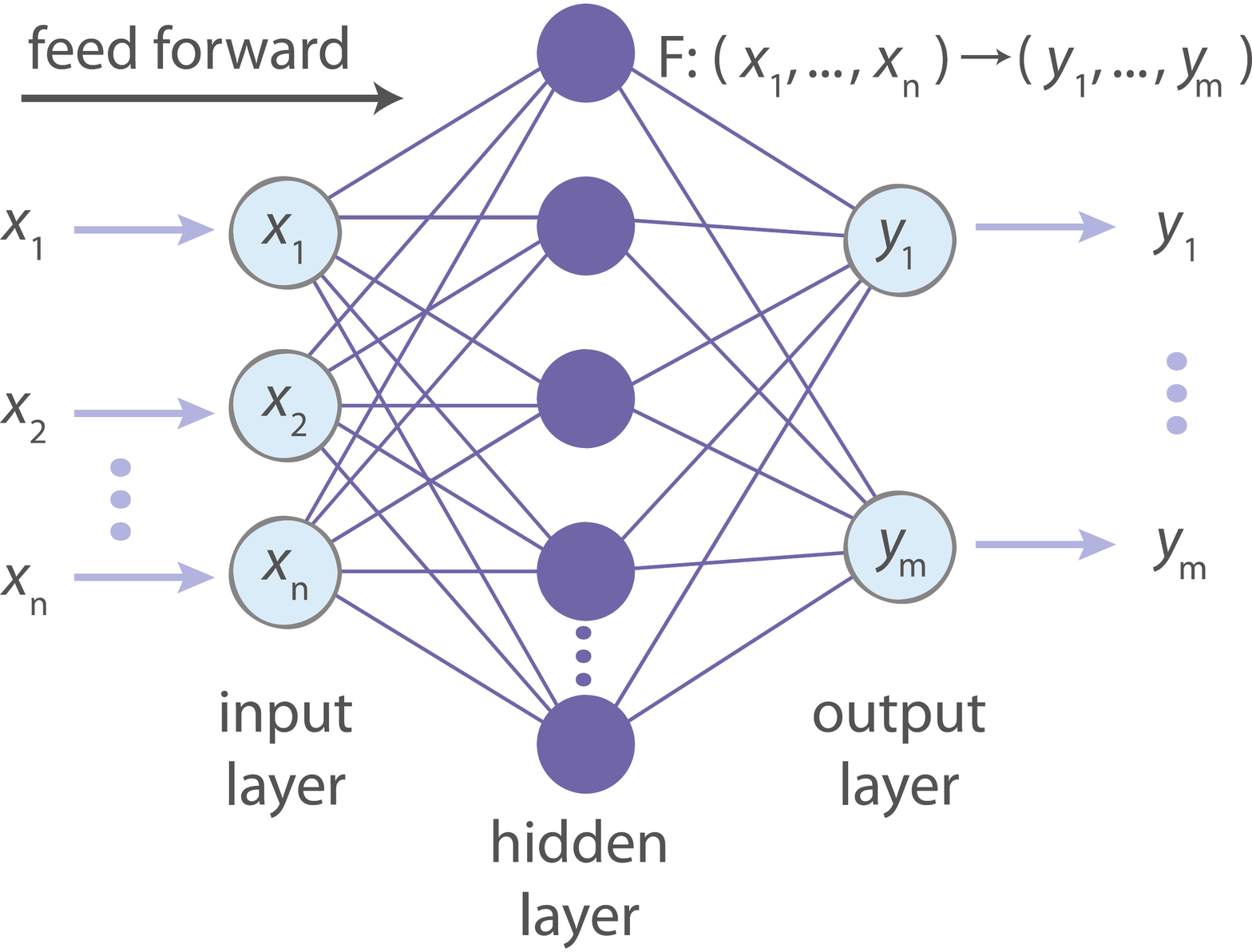}%
}
\caption{{\bf Neural networks. (a)} Single artificial neuron with weights $w_i$, bias $b$ and ELU activation function $\sigma_{\mathrm{ELU}}$. The inputs to the neuron are denoted by $x_1,\dots,x_4$.  {\bf (b)} ELU activation function for $\alpha=1$.   {\bf (c)} Fully connected (feed-forward) neural network with 3 layers. The network as a whole can be thought of as a function mapping the inputs $(x_1,\dots,x_n)$ to the output $(y_1,\dots,y_m)$.}
\end{figure*}

We can consider an autoencoder as a special case of \PN{}, where we ask always the same question and expect the network to reproduce the observation input. Hence, an autoencoder can be described by a triple $(\mathcal{O},\mathcal{Q}=\{0\},\lab:(o,0) \mapsto o)$.
As a corollary of Proposition~\ref{prop:latent_number}, we show that in the case of an autoencoder, the required number of latent variables corresponds to the ``relevant" number of degrees of freedom that describe the observation input. The relevant degrees of freedom, which are called {\it (hidden) generative factors} in this context in representation learning (see for example~\cite{Higgins2017}), may be described by the dimension of the domain of a smooth nondegenerate data generating function $H$, defined as follows.

\begin{definition}
	We say that a smooth function $H:\mathcal{G} \subset \mathbb{R}^d \rightarrow \mathbb{R}^r$ is nondegenerate if there exists an open subset $\mathcal{N}_d \subset \mathcal{G}$ such that the restriction $H|_{\mathcal{N}_d}:\mathcal{N}_d \rightarrow H(\mathcal{N}_d)$ of $H$ on $\mathcal{N}_d$ is a diffeomorphism. 
\end{definition}

One may think of $H$ as sending a small dimensional representation of the data onto a manifold in a high dimensional space of observations.
\begin{corollary}[Minimal representation for an autoencoder]
Let $H:\mathcal{G} \subset \mathbb{R}^d \rightarrow \mathcal{O} \subset \mathbb{R}^r$ be a smooth, nondegenerate and surjective (data generating) function, and let us assume that $\mathcal{G}$ is bounded. Then the minimal sufficient representation 
for data described by a triple $(\mathcal{O},\mathcal{Q}=\{0\},\lab:(o,0)\mapsto o)$ contains $d$ latent variables.
\begin{proof}
First, we show the existence of a (sufficient) representation with $d$ latent variables. We define the encoder mapping (and hence the representation) by $E:o\mapsto \mathrm{argmin}[H^{-1}(\{o\})] \in \mathcal{G}$, where the minimum takes into  account only the first vector entry.\footnote{Note that any element in $H^{-1}(\{o\})$ could be chosen.}  
We set the decoder equal to the smooth map  $H$. By noting that $D(E(o),0)=o$ for all $o \in \mathcal{O}$, this shows that  $d$ latent variables are sufficient.

Let us now show that there  cannot exist a representation with less than $d$ variables. By definition of a nondegenerate function $H$, there exists an open subset $\mathcal{N}_d \subset \mathcal{G}$ in $\mathbb{R}^d$ such that $H|_{\mathcal{N}_d}:\mathcal{N}_d \rightarrow H(\mathcal{N}_d)$ is a diffeomorphism. We define the function $f:o \in H(\mathcal{N}_d) \mapsto \lab(o,0) \in \mathcal{I}$, where $\mathcal{I}=H(\mathcal{N}_d)$. Since $f$ is the identity map and hence a diffeomorphism, the data described by the triple $(\mathcal{O},\mathcal{Q}=\{0\},\lab:(o,0)\mapsto o)$ has dimension at least  $d$. By Proposition~\ref{prop:latent_number}, we conclude that at least  $d$ latent variables are required.
\end{proof}
\end{corollary}

\section{Neural networks} \label{app:NN}

For a detailed introduction to artificial neural networks and deep learning, see for example~\cite{nielsenneural}. Here we give a very short overview of the basics.

\paragraph{Single artificial neuron.} The building blocks of neural networks are single neurons (Figure~\ref{fig:neuron}). We can think of a neuron as a map that takes several real inputs $x_1,\dots,x_n$ and provides an output $\sigma(\sum_i w_i x_i +b)$, according to an \textit{activation function}  $\sigma:\mathbb{R}\rightarrow\mathbb{R}$,  where the weights $w_i \in \mathbb{R}$ and the bias $b \in \mathbb{R}$ are tunable parameters. The output of the neuron is itself sometimes denoted by \textit{activation}, and there  are different possible choices for the activation function. 
For the implementation of the examples in this paper, we use the exponential linear unit (ELU)~\cite{clevert_fast_2015}, depicted in Figure~\ref{fig:ELU}. The ELU is defined for a parameter $\alpha>0$ as
$$
\sigma_{\mathrm{ELU}}(z)=
\begin{cases}
z & \text{for } z>0 \, ,\\
\alpha \left(e^z-1\right) & \text{for } z\leq 0\,.\\
\end{cases}
$$
\paragraph{Neural network.}  A (feed-forward) neural network is created by arranging neurons in layers and forwarding the outcomes of the neurons in the $i$-th layer to neurons in the $(i+1)$-th layer (see Figure~\ref{fig:feedforward_net}). The network as a whole can be viewed as a function $F:\mathbb{R}^n \rightarrow \mathbb{R}^m$ with $x_1,\dots,x_n$ corresponding to the activations of the neurons in the first layer (which is called \textit{input layer}). The activations of the input layer form the input for the second layer, which is a \textit{hidden layer} (since it is neither an input nor an output layer). In the case of a fully connected network, each neuron in the $(i+1)$-th layer receives the activations of all neurons in the $i$-th layer as input. The activations of the $m$ neurons in the last layer, which is called \textit{output layer}, are then interpreted as the output of the function $F$.
It can be shown that neural networks are  \textit{universal}, in the sense that any continuous function can be approximated arbitrarily well by a feedforward network with just one hidden layer by using sufficiently many hidden neurons. For a mathematical statement of the result, see~\cite{cybenko_approximation_1989, hornik_multilayer_1989}. A visualization is given in~\cite{nielsenneural}.

\paragraph{Training.} The weights and the biases of the neural network are not tuned by hand; instead, they are optimized using training samples, i.e., known input-output-pairs $(x, F^{\star}(x))$ of the function $F^{\star}$ that we would like to approximate. We may think of a neural network as a  class of functions $\{F_\theta\}_\theta$, parametrized by $\theta$, which contains the weights and biases of all the neurons in the network. A cost function $C\left(x, \theta \right)$ measures how close the output $F_\theta(x)$ of the network is to the desired output $F^{\star}(x)$ for an input $x$. For example, a common choice for the cost function is $C\left(x, \theta \right) = \norm{F^{\star}(x) - F_\theta(x)}_2^2$.

The weights and biases of a network are initialized at random~\cite{nielsenneural}. 
To then update the parameters $\theta$, the gradient $\vec{\nabla}_\theta C\left(x, \theta \right)$ is computed and averaged over all training samples $x$. Subsequently, $\theta$ is updated in the negative gradient direction --- hence the name {\textit {gradient descent}}.
In practice, the average of the gradient over all training samples is often replaced by an average over a smaller subset of training samples called a {\it mini-batch}; then, the algorithm is called {\textit stochastic gradient descent}.
The backpropagation algorithm is used to perform a gradient descent step efficiently (see~\cite{nielsenneural} for details).

\subsection{Variational autoencoders} \label{app:VAE}

\begin{figure*}[t]  
    \centering
    \includegraphics[width=.6\textwidth]{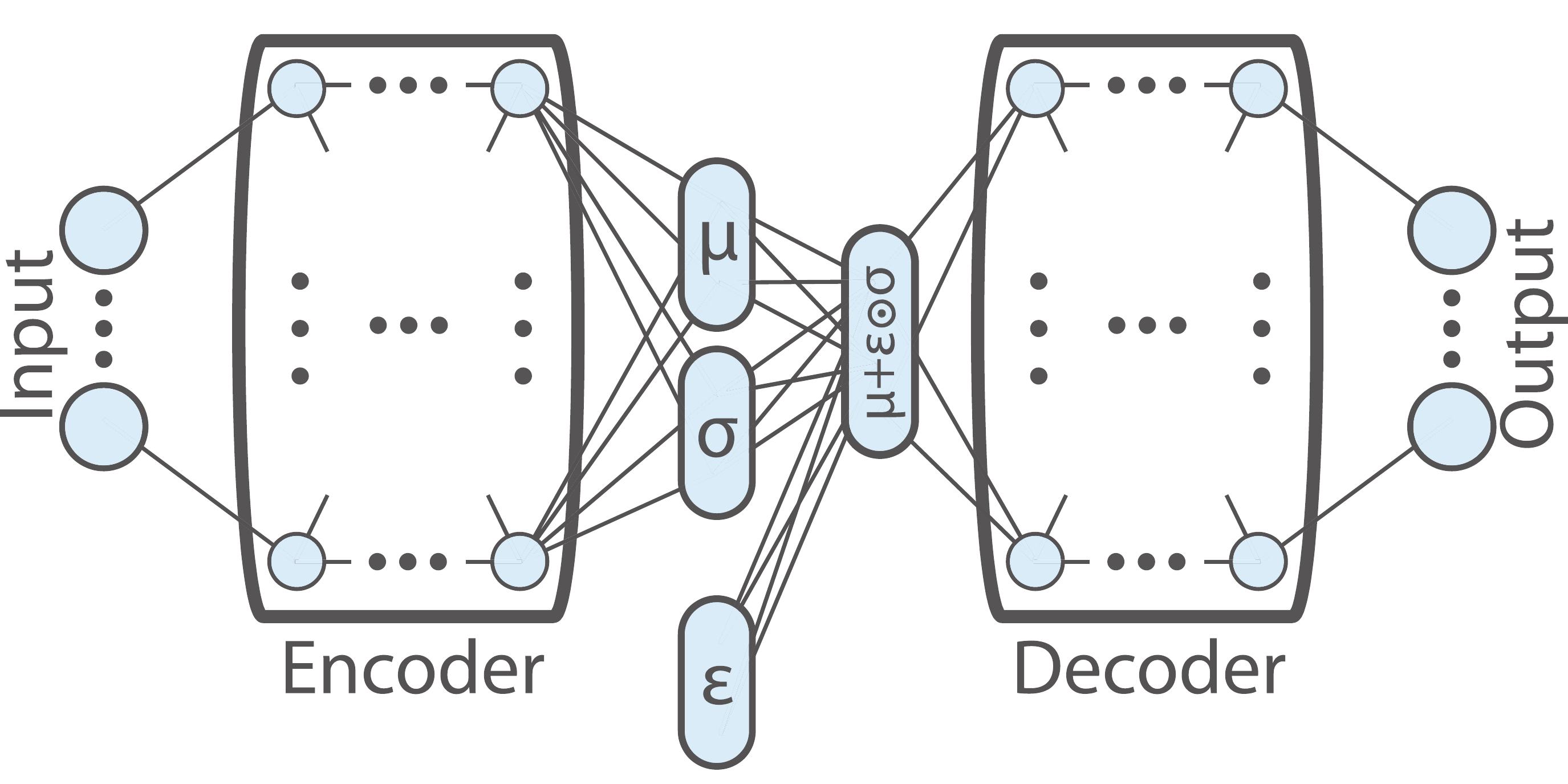}
   \caption{Network structure for a variational autoencoder.  The encoder and decoder are described by conditional probability distributions $p(z | x)$ and $p(x | z)$ respectively. The output distribution of the encoder are the parameters $\mu_i$ and $\log(\sigma_i)$ for independent Gaussian distributions $z_i \sim \Normal(\mu_i,\sigma_i)$ of the latent variables. The reparameterization trick is used to sample from the latent distribution. }
    \label{fig:VAE}
\end{figure*}

The implementation of \PN{} uses a modified version of so-called variational autoencoders (VAEs)~\cite{Kingma2013,Higgins2017}. The standard VAE architecture does not include the question input used by \PN{} and tries to reconstruct the input from the representation instead of answering a question. VAEs are one particular architecture used in the field of representation learning~\cite{Bengio2012}. Here, we give a short overview over the goals of representation learning and the details of VAEs.

\paragraph{Representation learning.} The goal in representation learning is to map a high-dimensional input vector $x$ to a lower-dimensional representation $z=(z_1,  z_2,\dots, z_d)$, commonly called the {\it latent vector}.\footnote{The variables $x$ and $z$ correspond to the observation $o$ and the representation $r$ used in the main text.}
The representation $z$ should still contain all the relevant information about $x$.
In the case of an autoencoder, $z$ is used to reconstruct the input $x$. This is motivated by the idea that the better the (low-dimensional) representation is, the better the original data can be recovered from it.
Specifically, an autoencoder uses a neural network (\emph{encoder}) to map the input $x$ to a small number of latent neurons $z$. Then, another neural network (\emph{decoder}) is used to reconstruct an estimate of the input, that is $z \mapsto \tilde x$. During training, the encoder and decoder are optimized to maximize the reconstruction accuracy and reach $\tilde x \approx x$. 

\paragraph{Probabilistic encoder and decoder.} Instead of considering deterministic maps $x \mapsto z$ and $z \mapsto \tilde x$, we generalize to conditional probability distributions $p(z | x)$ for the encoder and $p(\tilde x | z)$ for the decoder.
This is motivated by the Bayesian view that the most informative statement the encoder can output a description of a probability distribution over all latent vectors, instead of outputting a single estimate. The same reasoning holds  for the decoder. We use the notation $z \sim p(z)$ to indicate that $z$ is picked at random according to the distribution $p$.

We cannot treat the general case analytically, so we make restricting assumptions to simplify the setting. 
First we assume that the input can be perfectly compressed and reconstructed by an encoder and decoder which are both neural networks, that is 
we assume that the ideal distributions $p(z | x)$ and $p(\tilde x | z)$ that reach $x = \tilde x$ are members of parametric families $\{p_\phi(z | x)\}_\phi$ and $\{p_\theta(\tilde x | z)\}_\theta$, respectively. 
We further assume that it is possible to achieve this with a latent representation where each neuron is independent of the others, $p_\phi(z | x) = \prod_i p_\phi(z_i|x)$.
If these distributions turn out hard to find for a given dimension $d$ of the latent representation, we can try to increase the number of neurons of the representation to disentangle them. 
Finally, we make one more simplifying assumption, which is justified \textit{a posteriori} by good results: that we can reach a good approximation of  $p(z | x)$ by using only independent normal distributions for each latent neuron, $p_\phi(z_i|x) = \Normal(\mu_i, \sigma_i)$, where $\mu_i$ is the mean and $\sigma_i$ the variance. We can think of the encoder as mapping $x$ to the vectors $\mu = (\mu_1, \dots, \mu_d)$ and $\sigma= (\sigma_1, \dots, \sigma_d)$. 

The optimal settings for $\phi$ and $\theta$ are then learned as follows, see Figure~\ref{fig:VAE}:
\begin{enumerate}
\item The encoder with parameters (weights and biases) $\phi$ maps an input $x$ to $p_\phi(z | x) = \Normal[(\mu_1, \dots, \mu_d), (\sigma_1, \dots, \sigma_d)]$. 
\item A latent vector $z$ is sampled from $p_\phi(z | x)$.
\item The decoder with parameters (weights and biases) $\theta$ maps the latent vector $z$ to $p_\theta(\tilde x | z)$.
\item The parameters $\phi$ and $\theta$ are updated to maximize the likelihood of the original input $x$ under the decoder distribution $p_\theta(\tilde x | z)$.
\end{enumerate}

\paragraph{Reparameterization trick.} The operation that samples a latent vector $z$ from $p_\phi(z | x)$ is not differentiable with respect to the parameters $\phi$ and $\theta$ of the network. However, differentiability is necessary to train the network using stochastic gradient descent. This issue is solved by the reparameterization trick introduced in~\cite{Kingma2013}: if $p_\phi(z_i | x)$ is a Gaussian with mean $\mu_i$ and standard deviation $\sigma_i$, we can replace the sampling operation using an auxiliary random number $\epsilon_i \sim \Normal(0, 1)$. Then, a sample of the latent variable $z_i \sim \Normal(\mu_i, \sigma_i)$ can be generated by $z_i=\mu_i + \sigma_i \epsilon_i$. Sampling $\epsilon_i$ does not interfere with  the gradient descent because $\epsilon_i$ is independent of the trainable parameters $\phi$ and $\theta$. Alternatively, one can view this way of sampling as injecting noise into the latent layer~\cite{Achille2018}.

\paragraph{$\beta$-VAE cost function.}  A computationally tractable cost function for optimizing the parameters $\phi$ and $\theta$ was derived in~\cite{Kingma2013}. This cost function was extended in~\cite{Higgins2017} to encourage independency of the latent variables $z_1, \dots, z_d$ (or to encourage  ``disentangled'' representations, in the language of representation learning). The cost function in~\cite{Higgins2017} is known as the $\beta$-VAE cost function, 
\begin{align*}
	C_{\beta}(x)=  & -\Big[ \E_{z \sim p_{\phi}(z|x)} \ \log p_\theta( x | z)\  \Big] + \beta \  \DKL\left[p_\phi(z|x)\| h(z)\right] \, ,
\end{align*}
where the distribution $h(z)$ is a prior over the latent variables, typically chosen as the unit Gaussian\footnote{The interpretation of $h(z)$ as a prior is clear only when deriving VAEs as generative networks. For details, see~\cite{Kingma2013}.}, $\beta\geq0$ is a constant, and $\DKL$ is the Kullback-Leibler (KL) divergence, which is a quasi-distance\footnote{The KL divergence satisfies all axioms of a metric apart from symmetry.} measure between probability distributions, 	
\begin{align*}
		\DKL\left[p(z) \| q(z) \right] = \sum_z p(z)\log\left(\frac{p(z)}{q(z)}\right).
\end{align*}

Let us give an intuition for the motivation behind the $\beta$-VAE  cost function.
The first term is a log-likelihood factor, which encourages the network to recover the input data with high accuracy.  It asks ``for each $z$ , how likely are we to recover the original $x$ after the decoding?'' and takes the expectation of the logarithm of this likelihood $p_\theta(x|z)$ (other figures of merit could be used here in an alternative to the logarithm) over $z$ sampled from $p_\phi(z|x)$, in order to simulate the encoding.  In practice, this  expectation is often estimated with a single sample, which works well enough if the mini-batches are chosen sufficiently large~\cite{Kingma2013}. 

\begin{figure*}[t] 
\centering
\begin{examplebox}[label=box:pendulum]{Time evolution of a damped pendulum \hypersetup{colorlinks=true,linkcolor=white} (Section~\ref{sec:pendulum})}\hypersetup{colorlinks=true,linkcolor=black} 
\begin{description}
\item[Problem:]  Predict the position of a one-dimensional damped pendulum at different times.
 \item[Physical model:]
 Equation~of~motion: $ m\ddot{x} = -\kappa x-b\dot{x}\, .$ \\
 \hspace*{2.1cm} Solution: $x(t)=A_0 e^{-\frac{b}{2m}t} \cos(\omega t+\delta_0)$, with $\omega=\sqrt{\frac{\kappa}{m}}\sqrt{1-\frac{b^2}{4m\kappa}}\,.$
   \item[Observation:] Time series of positions: $o= \big[x(t_i)\big]_{i \in \{1,\dots,50\}} \in\mathbb{R}^{50}$, with equally spaced $t_i \in [0,5]s$. Mass $m=1$kg, amplitude $A_0=1$m and phase $\delta_0=0\,$ are fixed; spring constant $\kappa \in [5,10] \textnormal{ kg/s}^2$ and damping factor $b \in [0.5,1] \textnormal{ kg/s}$ are varied between training samples.
   	\item[Question:] Prediction times: $q=t_{\textnormal{pred}} \in [0,10]s$.
   \item[Correct answer:]
    Position at time $t_{\textnormal{pred}}$: $\lab=x(t_{\textnormal{pred}}) \in \mathbb{R}\,$.
   \item[Implementation:] Network depicted in Figure~\ref{fig:scinet_general} with three latent neurons (see Table~\ref{Table:training_details} and Table~\ref{Table:network_spec} for the details).
      \item[Key findings:] 
       \begin{itemize} \item[]
   \item \PN{} predicts the positions $x(t_{\textnormal{pred}})$ with a root mean square error below $2\%$ (with respect to the amplitude $A_0=1$m) (Figure~\ref{fig:pendulum_pred}).
  \item \PN{} stores $\kappa$ and $b$ in two of the latent neurons, and does not store any information in the third latent neuron (Figure~\ref{fig:pendulum_state_neurons}).
     \end{itemize}
    \end{description}
\end{examplebox}
\end{figure*}

The second term encourages disentangled representations, and we can motivate it using standard properties of the KL divergence. Our goal is to minimize the amount of correlations between the latent variables $z_i$: 
we can do this by minimizing the distance $\DKL\left[p(z) \| \prod_i p(z_i)\right] $ between $p(z)$ and the product of its marginals. 
For any other distribution with independent $z_i$,  $h(z)  = \prod_i h(z_i)$, the KL divergence satisfies
\begin{align}
	\DKL\left[p(z) \| \prod_i p(z_i)\right] &\leq \DKL\left[p(z) \| h(z)\right] \nonumber.
\end{align}
The KL divergence is furthermore jointly convex in its arguments, which implies
\begin{align*}
	\DKL& \left[\sum_x p(x)\  p_\theta(z|x) \| h(z) \right] \\ & \leq \sum_x p(x) \ \DKL\left[p_\theta(z|x) \| h(z) \right].
\end{align*}
Combining this with the previous inequality, we obtain
	\begin{align*}
		\DKL &\left[p(z) \  \| \prod_i p(z_i) \right]  \\ &\leq \E_{x\sim p(x)} \  \DKL\left[p(z|x) \| h(z)\right]. 
\end{align*}
The term on the right hand side corresponds exactly to the second term in the cost function, since in the training we try to minimize $\E_{x\sim p(x)} C_\beta (x)$. 
Choosing a large parameter $\beta$ also penalizes the size of latent representation $z$, motivating the network to learn an efficient representation.  
For an empirical test of the effect of large $\beta$ see~\cite{Higgins2017}, and for another theoretical justification using the information bottleneck approach see~\cite{Achille2018}.

\begin{figure*}
\centering
 \subfloat[ \label{fig:pendulum_pred}]{%
  \includegraphics[width=.4\linewidth]{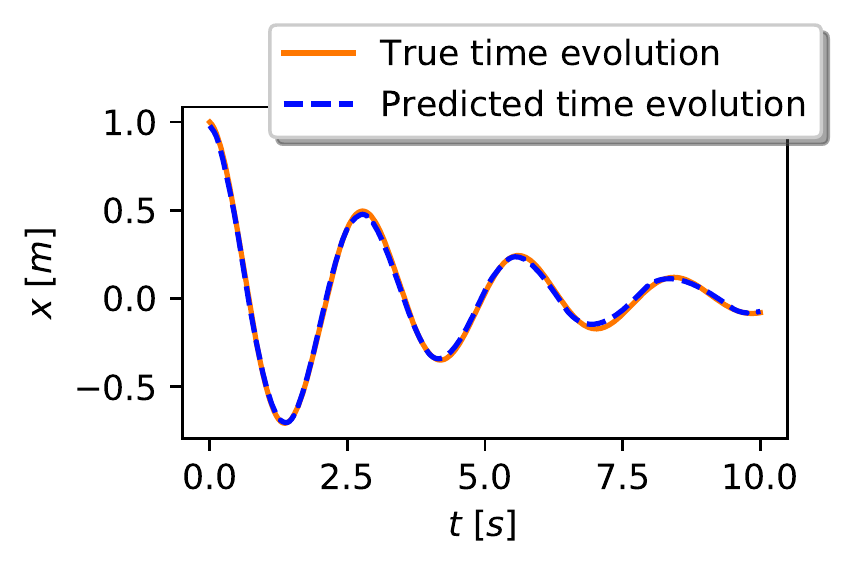}%
}\hfill
\subfloat[ \label{fig:pendulum_state_neurons}]{%
  \includegraphics[width=.4\linewidth]{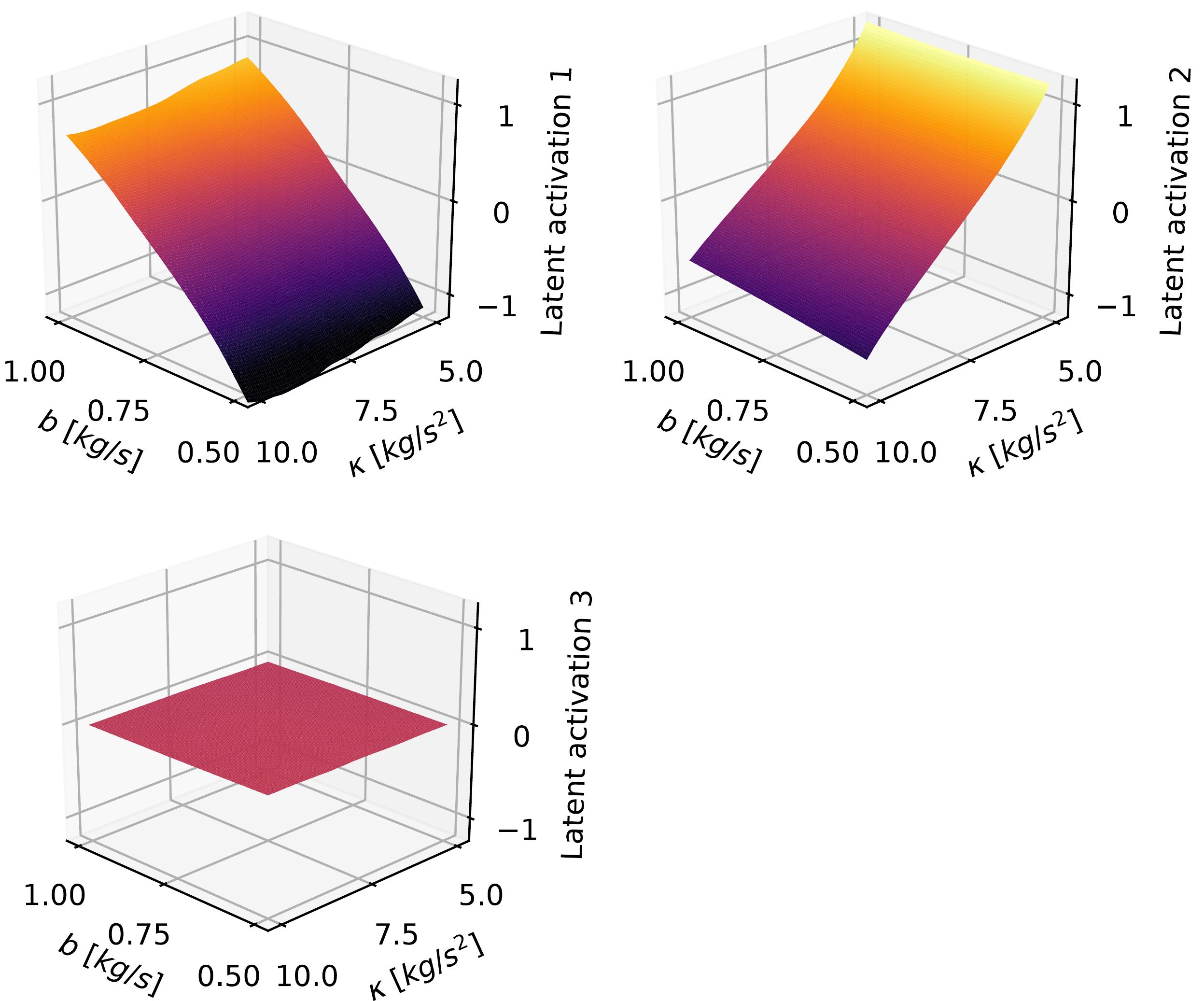}%
}
\caption{{\bf Damped pendulum}.  \PN{} is fed a time series of the trajectory of a damped pendulum. It learns to store the two relevant physical parameters, frequency and damping, in the representation, and makes correct predictions about the pendulum's future position.  {\bf (a) Trajectory prediction of \PN{}.} Here, the spring constant is $\kappa=5$kg/s$^2$ and the damping factor is $b=0.5$kg/s. \PN{}{\it's} prediction is in excellent agreement with the true time evolution. {\bf (b) Representation learned by \PN{}.} The plots show the activations of the three latent neurons of \PN{} as a function of the spring constant $\kappa$ and the damping factor $b$.
The first two neurons store the damping factor and spring constant, respectively.
The activation of the third neuron is close to zero, suggesting that only two physical variables are required. 
On an abstract level, learning that one activation can be set to a constant is encouraged by searching for uncorrelated latent variables, i.e., by minimizing the common information of the latent neurons during training.}
 \label{fig:damped_pendulum}
\end{figure*}

To derive  an explicit form of $C_\beta$ for a simple case, we again assume that $p_\phi(z|x) = \Normal(\mu, \sigma)$. In addition, we assume that the decoder output $p_\theta(\tilde x | z)$ is a multivariate Gaussian with mean $\hat{x}$ and fixed covariance matrix $\hat{\sigma} = \frac{1}{\sqrt{2}} \id$.
With these assumptions, the $\beta$-VAE cost function can be explicitly written as
\begin{align} \label{cost_function_betaVAE}
    C_\beta(x) = \lVert \hat{x} - x \rVert ^ 2_2 - \frac{\beta}{2}\left(\sum_{i} \log(\sigma_i^2)-\mu_i^2-\sigma_i^2\right) + C \, .
\end{align} 
The constant terms $C$ do not contribute to the gradients used for training and can therefore be ignored.

\section{Details about the physical examples \label{app:sec:examples}} 
In the following, we give some more information about the four examples of physical systems to which we applied \PN{} and which were mentioned in the main text.

\subsection{Damped pendulum} \label{sec:pendulum}

We consider a simple example from classical physics, the damped pendulum, described in Box~\ref{box:pendulum}.
The time evolution of the system is given by the differential equation  $-\kappa x-b\dot{x}=m\ddot{x}$, where $\kappa$ is the \textit{spring constant}, which determines the frequency of the oscillation, and $b$ is the \textit{damping factor}. We keep the mass $m$ constant (it is a scaling factor that could be absorbed by defining $\kappa\mathrlap{'}=\kappa/m$ and $b\mathrlap{'}=b/m$),
 such that $\kappa$ and $b$ are the only variable parameters. We consider the case of weak damping here, where the solution to the equation of motion is given in Box~\ref{box:pendulum}.

We choose a network structure for \PN{} with $3$ latent neurons. As an input, we provide a time series of positions of the pendulum and we ask \PN{} to predict the position at a future time (see Box~\ref{box:pendulum} for details).
The accuracy of the predictions given by \PN{} after training is illustrated in Figure~\ref{fig:pendulum_pred}.

Without being given any physical concepts, \PN{} learns to extract the two relevant physical parameters from (simulated) time series data for the $x$-coordinate of the pendulum and to store them in the latent representation. 
As shown in Figure~\ref{fig:pendulum_state_neurons}, the first latent neuron depends nearly linearly on $b$ and is almost independent of $\kappa$, and the second latent neuron depends only on $\kappa$, again almost linearly. Hence, \PN{} has recovered the same time-independent parameters $b$ and $\kappa$ that are used by physicists. The third latent neuron is nearly constant and does not provide any additional information --- in other words, \PN{} recognized that two parameters suffice to encode this situation.

\subsection{Conservation of angular momentum} \label{sec:cons_law}


One of the most important concepts in physics is that of conservation laws, such as conservation of energy and angular momentum.
While their relation to symmetries makes them interesting to physicists in their own right, conservation laws are also of practical importance.
If two systems interact in a complex way, we can use conservation laws to predict the behaviour of one system from the behaviour of the other, without studying the details of their interaction.
For certain types of questions, conserved quantities therefore act as a compressed representation of joint properties of several systems.

We consider the scattering experiment shown in Figure~\ref{fig:angular_momentum_cons} and described in Box~\ref{box:conservation}, where two point-like particles collide. 
Given the initial angular momentum of the two particles and the final trajectory of one of them, a physicist can predict the trajectory of the other using conservation of total angular momentum.

To see whether \PN{} makes use of angular momentum conservation in the same way as a physicist would do, we train it with (simulated)
experimental data as described in Box~\ref{box:conservation}, and add  Gaussian noise to show that the encoding and decoding are robust.
Indeed, \PN{} does exactly what a physicist would do and stores the total angular momentum in the latent representation (Figure~\ref{fig:ang_mom_cons_plots}).
This example shows that \PN{} can recover  conservation laws, and suggests that they emerge  naturally from compressing data and asking questions about joint properties of several systems.\\

\begin{figure*}[ht!] 
\centering
\begin{examplebox}[label=box:conservation]{Two-body collision with angular momentum conservation\hypersetup{colorlinks=true,linkcolor=white} (Section~\ref{sec:cons_law})}\hypersetup{colorlinks=true,linkcolor=black}
\begin{description}
\item[\textbf{Problem}:] Predict the position of a particle fixed on a rod of radius $r$ (rotating about the origin) after a collision at the point $(0,r)$ with a free particle (in two dimensions, see Figure \ref{fig:ang_mom_setup}).
\item[\textbf{Physical model}:] Given the total angular momentum before the collision and the velocity of the free particle after the collision, the position of the rotating particle at time $\toutPred$ (after the collision) can be calculated from angular momentum conservation: 
$J=m_{\rm rot} r^2 \omega-r m_{\rm free} ({\bf v}_{\rm free})_x=m_{\rm rot} r^2 \omega\mathrlap{'}-r m_{\rm free} ({\bf v}\mathrlap{'}_{\rm free})_x=J\mathrlap{'}$. 
  \item[Observation:]\hfill \\Time series of both particles before the collision: 
$o= [\left(\tin_i^{\textnormal{rot}},{\bf q}_{\rm rot}(\tin_i^{\textnormal{rot}})\right),\left(\tin_i^{\textnormal{free}},{\bf q}_{\rm free} (\tin_i^{\textnormal{free}})\right) ]_{i \in \{1,\dots,5 \}}$, with times $\tin_i^{\textnormal{rot}}$ and $\tin_i^{\textnormal{free}}$ randomly chosen for each training sample.
Masses $m_{\rm rot}=m_{\rm free}=1$kg and the orbital radius $r=1$m are fixed; 
initial angular velocity $\omega$, initial velocity $\mathbf v_{\rm free}$, the first component of ${\bf q}_{\rm free}(0)$ and final velocity $\mathbf v'_{\rm free}$ are varied between training samples.
Gaussian noise ($\mu = 0, \sigma = 0.01$m) is added to all position inputs.

\item[Question:] Prediction time and position of free particle after collision: $q=\left(\toutPred, \left[\tout_i,{\bf q}\mathrlap{'}_{\rm free}(\tout_i)\right]_{i \in \{ 1,\dots,5\}} \right)\, .$
	  
\item[Correct answer:] Position of rotating particle at time $\toutPred$: $\lab={\bf q}\mathrlap{'}_{\rm rot}(\toutPred)\, .$
 \item[Implementation:] Network depicted in Figure~\ref{fig:scinet_general} with one latent neuron (see Table~\ref{Table:training_details} and Table~\ref{Table:network_spec} for the details).
   \item[Key findings:] 
      \begin{itemize} \item[]
 \item \PN{} predicts the position of the rotating particle with root mean square prediction error below 4\% (with respect to the radius $r=1$m).
  \item \PN{} is resistant to noise.
  \item \PN{} stores the total angular momentum in the latent neuron.
    \end{itemize}
    \end{description}
\end{examplebox}
\end{figure*}

\subsection{Representation of qubits} \label{sec:rep_qubit}
 \textit{Quantum state tomography} is an active area of research~\cite{paris_quantum_2004}.
Ideally, we look for a \emph{faithful representation} of the state of a quantum system, such as the wave function: a representation that stores all information necessary to predict the probabilities of the outcomes for arbitrary measurements on that system.
However, to specify a faithful representation of a quantum system it is not necessary to perform all theoretically possible measurements on the system. If a set of measurements is sufficient to reconstruct the full quantum state, such a set is called \textit{tomographically complete}.

Here we show that, based only on (simulated) experimental data and without being given any assumptions about quantum theory, \PN{} recovers a faithful representation of the state of small quantum systems and can make accurate predictions.
In particular, this allows us to infer the dimension of the system and distinguish tomographically complete from incomplete measurement sets.
Box~\ref{box:qubit} summarizes the setting and the results.


\begin{figure*}[t] 
\centering
 \subfloat[ \label{fig:ang_mom_setup}]{%
  \includegraphics[width=.6\linewidth]{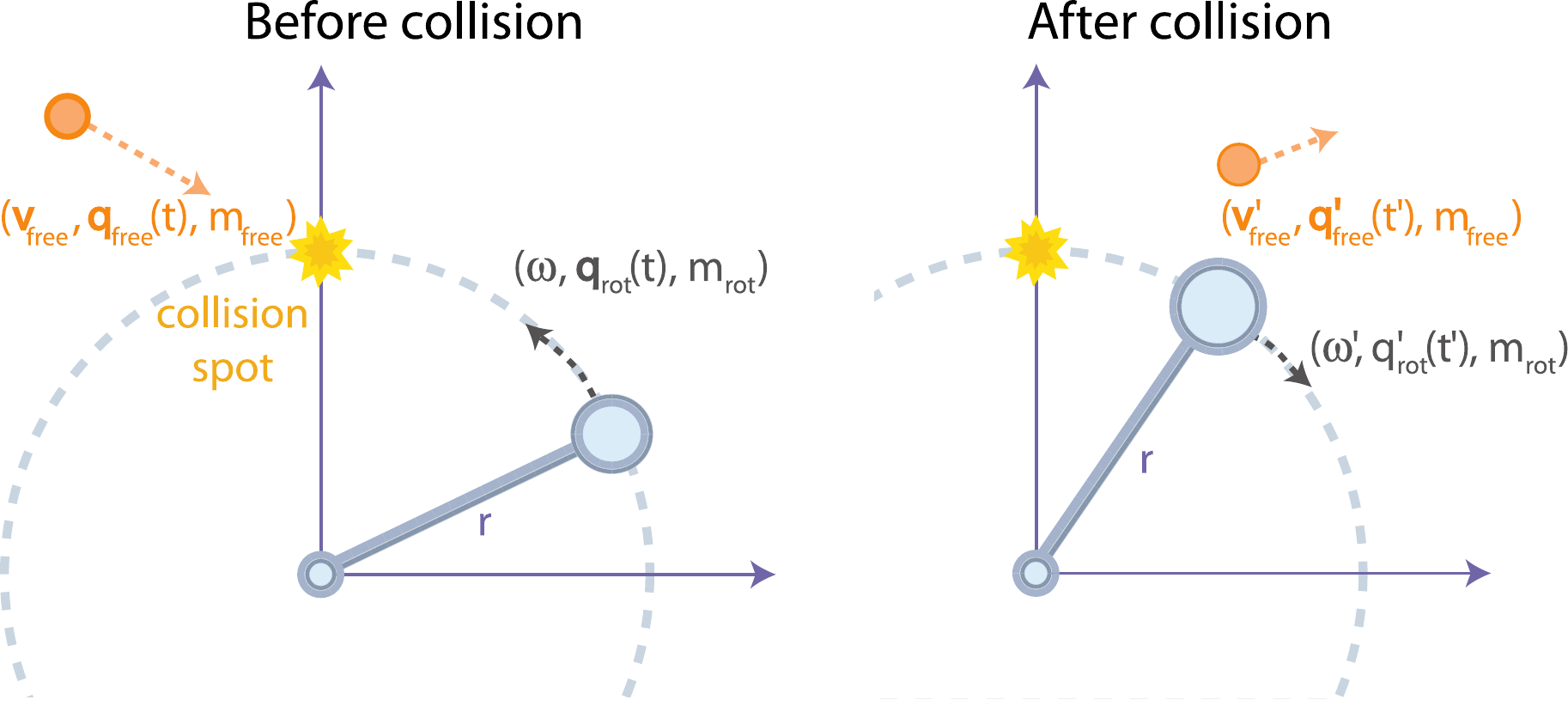}%
}\hfill
\subfloat[ \label{fig:ang_mom_cons_plots}]{%
  \includegraphics[width=.3\linewidth]{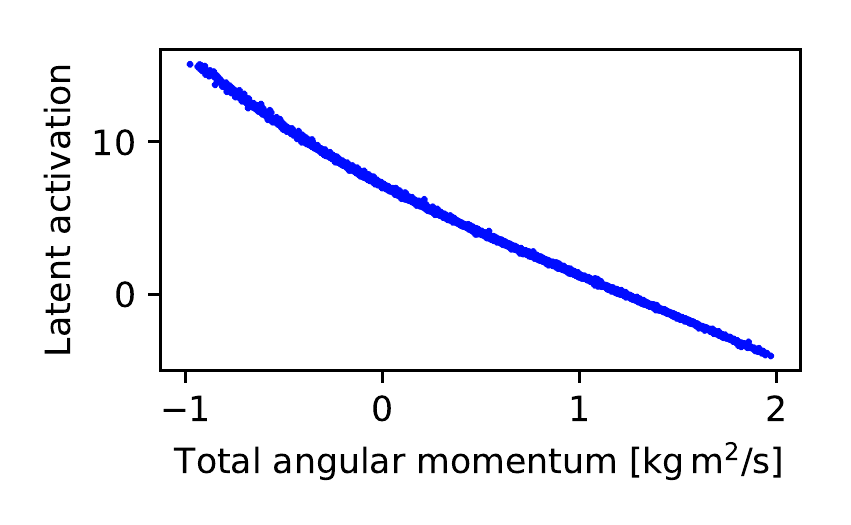}%
}
\caption{{\bf Collision under conservation of angular momentum.} In a classical mechanics scenario where the total angular momentum is conserved, the neural network learns to store this quantity in the latent representation. {\bf (a) Physical setting.} A body of mass $m_{\rm rot}$ is fixed on a rod of length $r$ (and of negligible mass) and rotates around the origin with angular velocity $\omega$.
A free particle with velocity ${\bf v}_{\rm free}$ and mass $m_{\rm free}$ collides with the rotating body at position ${\bf q} =(0,r)$.
After the collision, the angular velocity of the rotating particle is $\omega\mathrlap{'}$ and the free particle is deflected with velocity ${\bf v}\mathrlap{'}_{\rm free}$. {\bf (b) Representation learned by \PN{}.} Activation of the latent neuron as a function of the total angular momentum. \PN{} learns to store the total angular momentum, a conserved quantity of the system.}
 \label{fig:angular_momentum_cons}
\end{figure*}

\begin{figure*}[ht!]
\begin{examplebox}[label=box:qubit]{Representation of pure one- and two-qubit states \hypersetup{colorlinks=true,linkcolor=white} (Section~\ref{sec:rep_qubit})}\hypersetup{colorlinks=true,linkcolor=black}
\begin{description}
	\item[\textbf{Problem}:] Predict the measurement probabilities for any  binary projective measurement $\omega \in \mathbb{C}^{2^n}$ on a pure $n$-qubit state $\psi  \in \mathbb{C}^{2^n}$ for $n=1,2$.
 \item[\textbf{Physical model}:]  The probability $p(\omega,\psi)$ to measure 0 on the state $\psi \in \mathbb{C}^{2^n}$ performing the measurement $\omega \in \mathbb{C}^{2^n}$ is given by $p(\omega,\psi)=|\braket{\omega}{\psi}|^2\, .$
  \item[Observation:] Operational parameterization of a state $\psi$: $o= \left[p(\alpha_{i},\psi)\right]_{i \in \{1,\dots,n_1\}}$ for a fixed set of random binary projective measurements $\mathcal{M}_1:=\left\{\alpha_1,\dots,\alpha_{n_1}\right\}$ ($n_1 = 10$ for one qubit, $n_1 = 30$ for two qubits).
  \item[Question:] Operational$^{\textnormal{\ref{footnote:para}}}$ parameterization of a measurement $\omega$: $q=\left[ p(\beta_{i},\omega)\right]_{i \in \{1,\dots,n_2 \}}$ for a fixed set of random binary projective measurements $\mathcal{M}_2:=\left\{ \beta_{1},\dots,\beta_{n_2} \right\}$ ($n_2 = 10$ for one qubit, $n_2 = 30$ for two qubits). 
  \item[Correct answer:] $\lab(\omega,\psi)=p(\omega,\psi)=|\braket{\omega}{\psi}|^2$.
   \item[Implementation:] Network depicted in Figure~\ref{fig:scinet_general} with varying numbers of latent neurons  (see Table~\ref{Table:training_details} and Table~\ref{Table:network_spec} for the details).
    \item[Key findings:] 
      \begin{itemize} \item[]
 \item \PN{} can be used to determine the minimal number of parameters necessary to describe the state $\psi$ (see Figure~\ref{fig:qubits_tomography_detailed}) without being provided with any prior knowledge about quantum physics.
  \item \PN{} distinguishes tomographically complete and incomplete sets of measurements (see Figure~\ref{fig:qubits_tomography_detailed}).
    \end{itemize}
    \end{description}
\end{examplebox}
\end{figure*}

\begin{figure*}[ht!] 
\centering
  \includegraphics[width=0.8\textwidth]{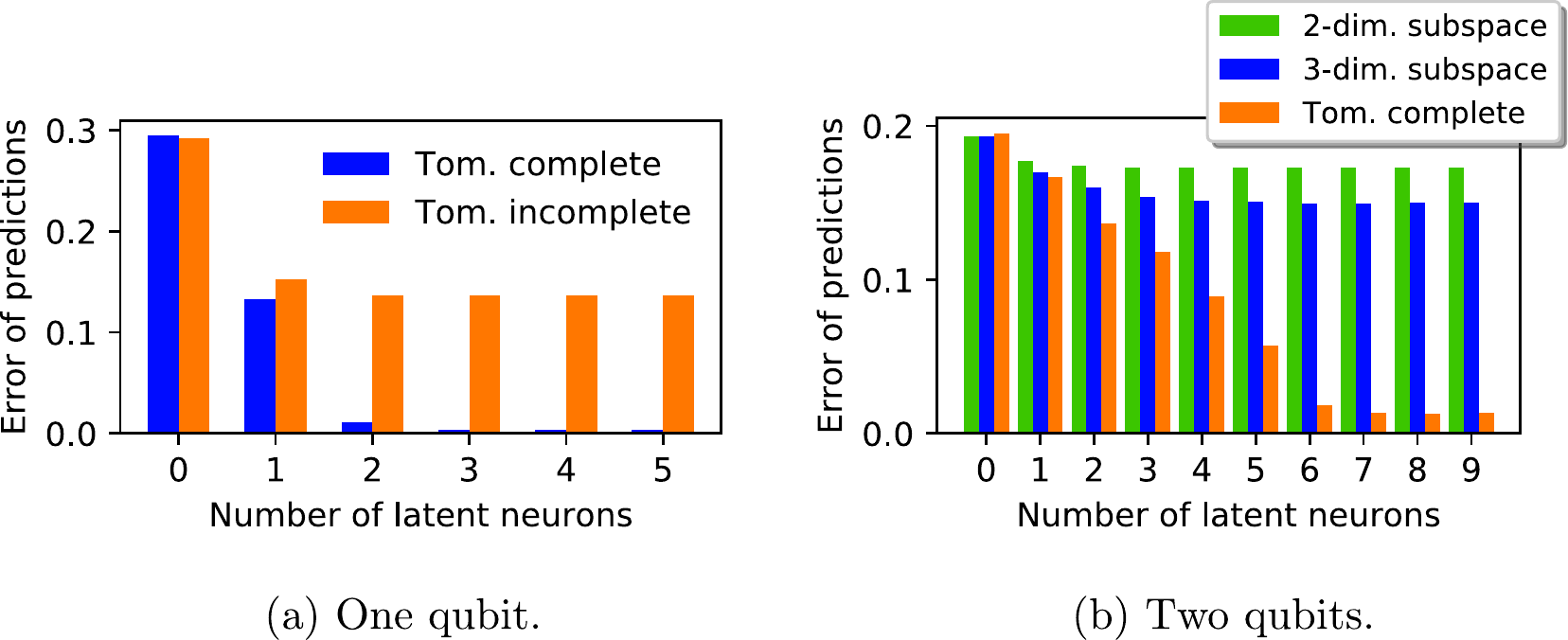}
\caption{{\bf Quantum tomography (more detailed plots than in the main text).}  \PN{} is given tomographic data for one or two qubits and an operational description of a measurement as a question input and has to predict the probabilities of outcomes for this measurement. We train \PN{} with both tomographically complete and incomplete sets of measurements, and find that, given tomographically complete data, \PN{} can be used to find the minimal number of parameters needed to describe a quantum state (two parameters for one qubit and six parameters for two qubits). Tomographically incomplete data can be recognized, since \PN{} cannot achieve perfect prediction accuracy in this case, and the prediction accuracy can serve as an estimate for the amount of information provided by the tomographically incomplete set. The plots show the root mean square error of \PN{}\textit{'s} measurement predictions for test data as a function of the number of latent neurons.}
\label{fig:qubits_tomography_detailed}
\end{figure*}

\begin{figure*}[ht!]
\begin{examplebox}[label=boxcopernicus]{Heliocentric model of the solar system\hypersetup{colorlinks=true,linkcolor=white} (Section~\ref{sec:copernicus})}\hypersetup{colorlinks=true,linkcolor=black}
\begin{description}
\item[\textbf{Problem}:] Predict the angles $\thetaM(t)$ and $\thetaS(t)$ of Mars and the Sun as seen from Earth, given initial states $\thetaM(t_0)$ and $\thetaS(t_0)$.
   \item[\textbf{Physical model}:] Earth and Mars orbit the Sun with constant angular velocity on (approximately) circular orbits.
  \item[Observation:] Initial angles of Mars and the Sun as seen from Earth: $o= \left(\thetaM(t_0),\thetaS(t_0)\right)$, randomly chosen from a set of weekly (simulated) observations within Copernicus' lifetime (3665 observations in total).
	\item[Question:] Implicit.   
  \item[Correct answer:] Time series $\big[a(t_1),\dots,a(t_n)\big]= \big[(\thetaM(t_1),\thetaS(t_1)), \dots, (\thetaM(t_n),\thetaS(t_n))\big]$ of $n=20$ (later in training: $n=50$) observations, with time steps $t_{i+1}-t_i$ of one week.
    \item[Implementation:]  Network depicted in Figure~\ref{fig:NN_structure_for_dynamic_variables} with two latent neurons and allowing for time updates of the form $r(t_{i+1}) = r(t_i) + b$  (see Table~\ref{Table:training_details} and Table~\ref{Table:network_spec} for the details).  
      \item[Key findings:] 
      \begin{itemize}\item[]
 \item \PN{} predicts the angles of Mars and the Sun with a root mean square error below 0.4\% (with respect to $2\pi$).
  \item \PN{} stores the angles $\phiE$ and $\phiM$ of the Earth and Mars as seen from the Sun in the two latent neurons (see Figure~\ref{fig:copernicus_phi_plot}). 
    \end{itemize}
    \end{description}
\end{examplebox}
\end{figure*}

A (pure) state on $n$ qubits can be  represented by a normalized complex vector $\psi\in \mathbb{C}^{2^n}$, where two states $\psi$ and $\psi'$ are identified if  and only if they differ by a global phase factor, i.e., if there exists $\phi \in \mathbb{R}$ such that  $\psi= e^{i \phi}\psi'$.
The normalization condition and irrelevance of the global phase factor decrease the number of free parameters of a quantum state by two.
Since a complex number has two real parameters, a single-qubit state is described by $2 \times 2^1-2=2$ real parameters, and a state of two qubits is described by $2 \times 2^2-2=6$ real parameters.

Here, we consider binary projective measurements on $n$ qubits. Like states, these measurements can be described by vectors $\omega \in \mathbb{C}^{2^n}$, with measurement outcomes labeled by $0$ for the projection on $\omega$ and $1$ otherwise.
The probability to get outcome $0$ when measuring  $\omega$ on a quantum system in state $\psi$ is then given by $p(\omega,\psi)=|\braket{\omega}{\psi}|^2$, where $\braket{\cdot}{\cdot}$ denotes the standard scalar product on $\mathbb{C}^{2^n}$.

To generate the training data for \PN{}, we assume that we have one or two qubits in a lab that can be prepared in arbitrary states and we have the ability to perform binary projective measurements in a set $\mathcal{M}$.
We choose $n_1$ measurements $\mathcal{M}_1:=\left\{\alpha_1,\dots,\alpha_{n_1}\right\} \subset\mathcal{ M}$ randomly, which we would like to use to determine the state of the quantum system. We perform all measurements in $\mathcal{M}_1$ several times on the same quantum state $\psi$ to estimate the probabilities $p(\alpha_i,\psi)$ of measuring $0$ for the $i$-th measurement.
These probabilities form the observation given to \PN{}.

To parameterize the measurement $\omega$, whose outcome probabilities should be predicted by \PN{}, we  choose another random set of measurements $\mathcal{M}_2:=\left\{ \beta_{1},\dots,\beta_{n_2} \right\} \subset \mathcal{M}$.
The probabilities $p(\beta_i,\omega)$ are provided to \PN{} as the question input. 
We always assume that we have chosen enough measurements in $\mathcal{M}_2$ such that they can distinguish all the possible measurements $\omega \in \mathcal{M}$, i.e., we assume that $\mathcal{M}_2$ is tomographically complete.\footnote{\label{footnote:para}This parameterization of a measurement $\omega$ assumes that we know the equivalence between binary projective measurements and states.
However, this is not a fundamental assumption, since we could parameterize the set of possible measurements by any parameterization that is natural for the experimental setup, for example the settings of the dials and buttons on an experimental apparatus.
Such a natural parameterization is assumed to fully specify the measurement, in the sense that the same settings on the experimental apparatus will always result in the same measurement being performed.
Because $\mathcal{M}_2$ only represents our choice for parameterizing the measurement setup, it is natural to assume that $\mathcal{M}_2$ is tomographically complete.} \PN{} then has to predict the probability $p(\omega,\psi)$ for measuring the outcome 0 on the state $\psi$ when performing the measurement $\omega$.

We train \PN{} with different pairs $(\omega, \psi)$ for one and two qubits, keeping the measurement sets $\mathcal{M}_1$ and $\mathcal{M}_2$ fixed. We choose $n_1 = n_2 = 10$ for the single-qubit case and $n_1 = n_2 = 30$ for the two-qubit case.
The results are shown in Figure~\ref{fig:qubits_tomography_detailed}.

Varying the number of latent neurons, we can observe how the quality of the predictions improves as we allow for more parameters in the representation of $\psi$.
To minimize statistical fluctuations due to the randomized initialization of the network, each network specification is trained three times and the run with the lowest  mean square prediction error on the test data is used.

For the cases where $\mathcal{M}_1$ is tomographically complete, the plots in Figure~\ref{fig:qubits_tomography_detailed} show a drop in prediction error when the number of latent neurons is increased up to two or six for the cases of one and two qubits, respectively.\footnote{In the case of a single qubit, there is an additional small improvement in going from two to three latent neurons: 
this is a technical issue caused by the fact that any two-parameter representation of a single qubit, for example the Bloch sphere representation, includes a \emph{cyclic parameter}, which cannot be exactly represented by a continuous encoder (see Appendix~\ref{app:cyclic_rep}). The same likely applies in the case of two qubits, going from 6 to 7 latent neurons. This restriction also makes it difficult to interpret the details of the learned representation.} 
This is in accordance with the number of parameters required to describe a one- or a two-qubit state. Thus, \PN{} allows us to extract the dimension of the underlying quantum system from tomographically complete measurement data, without any prior information about quantum mechanics.

\PN{} can also be used to determine whether the measurement set $\mathcal{M}_1$ is tomographically complete or not.
To generate tomographically incomplete data, we choose the measurements in $\mathcal{M}_1$ randomly from a subset of all binary projective measurements. 
Specifically, the quantum states corresponding to measurements in $\mathcal{M}_1$ are restricted to random \emph{real} linear superpositions of $k$ orthogonal states, i.e., to a (real) $k$-dimensional subspace.
For a single qubit, we use a two-dimensional subspace; for two quibts, we consider both two- and three-dimensional subspaces.

Given tomographically incomplete data about a state $\psi$, it is not possible for \PN{} to predict the outcome of the final measurement perfectly regardless of the number of latent neurons, in contrast to the tomographically complete case (see Figure~\ref{fig:qubits_tomography_detailed}).
Hence, we can deduce from \PN{}{\it 's} output that $\mathcal{M}_1$ is an incomplete set of measurements.
Furthermore, this analysis provides a qualitative measure for the amount of information provided by the tomographically incomplete measurements: in the two-qubit case, increasing the subspace dimension from two to three leads to higher prediction accuracy and the required number of latent neurons increases.

\subsection{Heliocentric model of the solar system}  \label{sec:copernicus}
All the details about this example were already given in the main text and are summarized in Box~\ref{boxcopernicus}.

\section{Representations of cyclic parameters} \label{app:cyclic_rep}
Here we explain the difficulty of a neural network to learn representations of cyclic parameters, which was alluded to in the context of the qubit example  (Section~\ref{sec:rep_qubit}, see~\cite{pitelis_learning_2013,korman_autoencoding_2018} for a detailed discussion relevant to computer vision).
In general, this problem occurs if the data $\mathcal{O}$ that we would like to represent forms a closed manifold (i.e., a compact manifold without boundary), such as a circle, a sphere or a Klein bottle.
In that case, several coordinate charts are required to describe this manifold. 

As an example, let us consider data points lying on the unit sphere $\mathcal{O}=\{(x,y,z): x^2+y^2+z^2=1\}$, which we would like to encode into a simple representation. The data can be (globally) parameterized with spherical coordinates $\phi \in [0,2 \pi)$ and $\theta \in [0,\pi]$ where  $(x,y,z)=f(\theta,\phi) \coloneqq (\sin \theta \cos \phi, \sin \theta \sin \phi,\cos \theta)$.\footnote{The function $f$ is not a chart, since it is not injective and its domain is not open.} 
We would like the encoder to perform the mapping $f^{-1}$, where we define $f^{-1}((0,0,1))=(0,0)$ and $f^{-1}((0,0,-1))=(\pi,0)$ for convenience. This mapping  is not continuous at points on the sphere with $\phi=0$ for $\theta \in (0,\pi)$.
Therefore, using a neural network as an encoder leads to problems, as neural networks, as introduced here, can only implement continuous functions.
In practice, the network is forced to approximate the discontinuity in the encoder by a very steep continuous function, which leads to a high error for points close to the discontinuity.

In the qubit example, the same problem appears. 
To parameterize a qubit state $\psi$ with two parameters, the Bloch sphere with parameters $\theta \in [0,\pi]$ and $\phi \in [0,2\pi)$ is used: the state $\psi$ can be written as $\psi(\theta,\phi) = (\cos(\theta/2),e^{i\phi} \sin(\theta/2))$ (see for example~\cite{nielsen_quantum_2010} for more details).
Ideally, the encoder would perform the map $E:o(\psi(\theta,\phi)):=\left(|\braket{\alpha_1}{\psi(\theta,\phi)}|^2,\dots,|\braket{\alpha_{N_1}}{\psi(\theta,\phi)}|^2\right) \mapsto (\theta,\phi)$ for some fixed binary projective measurements $\alpha_i \in \mathbb{C}^2$. 
However, such an encoder is not continuous. Indeed, assuming that the encoder is continuous, leads to the following contradiction:
\newline
\begin{align*}
(\theta,0) 
&= E(o(\psi(\theta,\phi = 0))) \\
&= E(o(\lim_{\phi \to 2 \pi} \psi(\theta,\phi))) \\
&= \lim_{\phi \to 2 \pi} E(o(\psi(\theta,\phi))) \\
&= \lim_{\phi \to 2 \pi} (\theta, \phi) = (\theta, 2 \pi) \,,
\end{align*}
where we have used the periodicity of $\phi$ in the second equality and the fact that the Bloch sphere representation and the scalar product (and hence $o(\psi(\theta, \phi))$) as well as the encoder (by assumption) are continuous in $\phi$ in the third equality.

\bibliographystyle{apsrev4-1}
\bibliography{PhysRepBib.bib}

\end{document}